\def\year{2021}\relax
\documentclass[letterpaper]{article} 
\usepackage{aaai21}  
\usepackage{times}  
\usepackage{helvet} 
\usepackage{courier}  
\usepackage[hyphens]{url}  
\usepackage{graphicx} 
\usepackage{natbib}  
\usepackage{caption} 
\nocopyright
\title{Margin of Victory in Tournaments: Structural and Experimental Results}
\author{
    Markus Brill,\textsuperscript{\rm 1}
    Ulrike Schmidt-Kraepelin,\textsuperscript{\rm 1}
    Warut Suksompong\textsuperscript{\rm 2} \\
}
\affiliations{
    \textsuperscript{\rm 1} Research Group Efficient Algorithms, TU Berlin, Germany \\
    \textsuperscript{\rm 2} School of Computing, National University of Singapore, Singapore
}

\usepackage{url}
\usepackage[utf8]{inputenc}
\usepackage{subfig}
\usepackage{dblfloatfix}
\usepackage{amsmath}
\usepackage{amsthm}
\usepackage{booktabs}
\usepackage{xspace}
\usepackage{cleveref}
\usepackage{tabularx}
\usepackage{amssymb}
\usepackage{thmtools}
\usepackage{thm-restate}
\usepackage{multibib}
\newcites{Appendix}{References}

\usepackage{enumerate}

\usepackage{pifont}
\newcommand{\cmark}{\ding{51}}%
\newcommand{\xmark}{\ding{55}}%

\usepackage{pbox}

\usepackage[cmyk,x11names]{xcolor}
\usepackage{tikz}
\usetikzlibrary{arrows.meta}
\tikzset{>={Latex[width=2mm,length=2mm]}}
\usetikzlibrary{decorations.pathmorphing,shapes,snakes}
\usetikzlibrary{calc}
\usetikzlibrary{arrows, decorations.markings}

\newtheorem{theorem}{Theorem}
\newtheorem{corollary}[theorem]{Corollary}
\newtheorem{proposition}[theorem]{Proposition}
\newtheorem{lemma}[theorem]{Lemma}

\theoremstyle{definition}
\newtheorem{definition}[theorem]{Definition}

\newtheorem*{claim}{Claim}

\newcommand{\cp}{\ensuremath{\mathit{CO}}\xspace}
\newcommand{\tc}{\ensuremath{\mathit{TC}}\xspace}
\newcommand{\uc}{\ensuremath{\mathit{UC}}\xspace}
\newcommand{\ba}{\ensuremath{\mathit{BA}}\xspace}
\newcommand{\E}{\mathbb{E}}

\newcommand{\mov}{\ensuremath{\mathsf{MoV}}}
\newcommand{\crs}{CRS\xspace}
\newcommand{\drs}{DRS\xspace}
\newcommand{\ie}{i.e.,\xspace}

\newcommand{\midd}{\mathrel{:}}

\DeclareMathOperator{\outdeg}{outdeg}
\DeclareMathOperator{\indeg}{indeg}
\DeclareMathOperator{\mincut}{min-cut}

\newcommand{\yes}{\textcolor{green!50!black}{\cmark}}
\newcommand{\no}{\textcolor{red!50!black}{\xmark}}

\newcommand{\secref}[1]{Section~\ref{#1}}

\usepackage[textsize=tiny,textwidth=1.5cm]{todonotes}

\newcommand{\citemov}{Brill et al. \shortcite{BrillScSu20}\xspace}
\newcommand{\citemovfull}{(Brill et al. 2020)\xspace}

\newcounter{int}
\makeatletter
\newcommand{\citen}[1] {\setcounter{int}{0}\@for\tmp:=#1\do{%
\ifnum \value{int}>0; \fi%
\setcounter{int}{1}%
\citeauthor{\tmp} \shortcite{\tmp}}}
\makeatother

\makeatletter
\newcommand{\citenp}[1]{\setcounter{int}{0}\@for\tmp:=#1\do{%
\ifnum \value{int}>0; \fi%
\setcounter{int}{1}%
\citeauthor{\tmp}, \citeyear{\tmp}}}
\makeatother

\allowdisplaybreaks

\begin{document}

\maketitle

\begin{abstract}
Tournament solutions are standard tools for identifying winners based on pairwise comparisons between competing alternatives. The recently studied notion of \textit{margin of victory (MoV)} offers a general method for refining the winner set of any given tournament solution, thereby increasing the discriminative power of the solution. 
In this paper, we reveal a number of structural insights on the MoV by investigating fundamental properties such as monotonicity and consistency with respect to the covering relation. 
Furthermore, we provide experimental evidence on the extent to which the MoV notion refines winner sets in tournaments generated according to various stochastic models. 
\end{abstract}

\section{Introduction}

Tournaments serve as a practical tool for modeling scenarios involving a set of alternatives along with pairwise comparisons between them.
Perhaps the most common example of a tournament is a round-robin sports competition, where every pair of teams play each other once and there is no tie in match outcomes.
Another application, typical especially in the social choice literature, concerns elections: here, alternatives represent election candidates, and pairwise comparisons capture the majority relation between pairs of candidates.
In order to select the ``winners'' of a tournament in a consistent manner, numerous methods---known as \emph{tournament solutions}---have been proposed.
Given the ubiquity of tournaments, it is hardly surprising that tournament solutions have drawn substantial interest from researchers in the past few decades~\citep{Laslier97,Woeginger03,Hudry09,AzizBrFi15,Dey17a,BrandtBrHa18,BrandtBrSe18,HanVa19}. 

While tournament solutions are useful for selecting the best alternatives according to various desiderata, several solutions suffer from the setback that they tend to choose large winner sets.
For instance, \citet{Fey08} showed that the top cycle, the uncovered set, and the Banks set are likely to include all alternatives in a large random tournament.
To address this issue, we recently introduced the notion of \emph{margin of victory (MoV)} for tournaments \citep{BrillScSu20}.
The MoV of a winner is defined as the minimum number of pairwise comparisons that need to be reversed in order for the winner to drop out of the winner set.
Analogously, the MoV of a non-winner is defined as the negative of the minimum number of comparisons that must be reversed for it to enter the winner set.\footnote{In our previous paper \citemovfull, we considered a more general setting where each pairwise comparison can have a weight representing the cost of reversing it, but here we will focus on the unweighted setting.}
In addition to refining tournament solutions, the notion also has a natural interpretation in terms of bribery and manipulation: the MoV of an alternative reflects the cost of bribing voters or manipulating match outcomes so that the status of the alternative changes with respect to the winner set.
For a number of common tournament solutions, we studied the complexity of computing the MoV and provided bounds on its values for both winners and non-winners \citemovfull.

Our previous results paint an initial picture on the properties of the MoV in tournaments.
Nevertheless, several important questions about the notion remain unanswered from that work.
For each tournament solution, how many different values does the MoV take on average?
How large is the set of alternatives with the highest MoV in a random tournament?
If two alternatives dominate the same number of other alternatives, for which tournament solutions is it the case that the MoV of both alternatives must be equal?
If an alternative ``covers'' another (i.e., the former alternative dominates the latter along with all alternatives that the latter dominates), for which tournament solutions is it always true that the MoV of the former alternative is at least that of the latter?
In this paper, we provide a comprehensive answer to these questions for several common tournament solutions using axiomatic and probabilistic analysis (\Cref{sec:structural}) as well as through experiments (\Cref{sec:experiments}).

\subsection{Related Work}

Despite their origins in social choice theory, tournament solutions have found applications in a wide range of areas including game theory~\citep{FisherRy95}, webpage ranking~\citep{BrandtFi07}, dueling bandit problems~\citep{RamamohanRaAg16}, and philosophical decision theory \citep{Podg20a}. 
As is the case for social choice theory in general, early studies of tournament solutions were primarily based on the axiomatic approach.
With the rise of computational social choice in the past fifteen years or so, tournament solutions have also been thoroughly examined from an algorithmic perspective.
For an overview of the literature, we refer to the surveys of \citet{Laslier97}, \citet{Hudry09}, and \citet{BrandtBrHa16}.

While we introduced the MoV concept for tournament solutions \citemovfull, similar concepts have been applied to a large number of settings, perhaps most notably voting.
In addition, various forms of bribery and manipulation have been considered for both elections and sports tournaments. 
We refer to our previous paper for relevant references, but note here that the MoV continues to be a popular concept in recent work, for example in the context of sports modeling~\citep{Kovalchik20}, election control~\citep{CastiglioniFeGa20}, and political and educational districting~\citep{StoicaChDe20}.
\citet{YangGu17} gave a parameterized complexity result for the decision version of computing the MoV with respect to the uncovered set.

The discriminative power of tournament solutions has been studied both analytically and experimentally.
As we mentioned earlier, 
\citet{Fey08} showed that in a large tournament drawn uniformly at random, the top cycle, the uncovered set, and the Banks set are unlikely to exclude any alternative.
\citet{ScottFe12} established an analogous result for the minimal covering set, while \citet{FisherRy95} proved that the bipartisan set selects half of the alternatives on average.
\citet{SaileSu20} extended some of these results to more general probability distributions, and \citet{BrandtSe16} performed experiments using both stochastic models and empirical data.

\section{Preliminaries}
\label{sec:prelims}

A \emph{tournament} $T=(V,E)$ is a directed graph in which exactly one directed edge exists between any pair of vertices. 
The vertices of $T$, denoted by $V(T)$, are often referred to as \textit{alternatives}, and their number $n:=|V(T)|$ is referred to as the \textit{size} of $T$. 
The set of directed edges of $T$, denoted by $E(T)$, represents an asymmetric and connex \emph{dominance relation} between the alternatives. 
An alternative $x$ is said to \emph{dominate} another alternative $y$ if $(x,y)\in E(T)$ (i.e., there is a directed edge from~$x$ to~$y$). 
When the tournament is clear from the context, we often write $x \succ y$ to denote $(x,y) \in E(T)$. By definition, for each pair $x,y$ of distinct alternatives, either $x$ dominates $y$ ($x \succ y$) or $y$ dominates $x$ ($y \succ x$), but not both. 
The dominance relation can be extended to sets of alternatives by writing $X\succ Y$ if $x\succ y$ for all $x\in X$ and all $y\in Y$. 

For a given tournament $T$ and an alternative $x \in V(T)$,
the \emph{dominion} of $x$, denoted by $D(x)$, is the set of alternatives~$y$ such that $x\succ y$.
Similarly, the set of \emph{dominators} of $x$, denoted by $\overline{D}(x)$, is the set of alternatives $y$ such that $y\succ x$.
The \emph{outdegree} of $x$ is denoted by $\outdeg(x) = |D(x)|$, and the \emph{indegree} of $x$ by $\indeg(x) = |\overline{D}(x)|$.
For any $x\in V(T)$, it holds that $\outdeg(x)+\indeg(x) = n-1$.
An alternative $x \in V(T)$ is said to be a \emph{Condorcet winner} in $T$ if it dominates every other alternative (i.e., $\outdeg(x) = n-1$), and a \emph{Condorcet loser} in $T$ if it is dominated by every other alternative (i.e., $\outdeg(x)=0$). 
A tournament is \emph{regular} if all alternatives have the same outdegree.
A regular tournament exists for every odd size, but not for any even size. 

\subsection{Tournament Solutions}

A \emph{tournament solution} is a function that maps each tournament to a nonempty subset of its alternatives, usually referred to as the set of \emph{winners} or the \emph{choice set}.
A tournament solution must not distinguish between isomorphic tournaments; in particular, if there is an automorphism that maps an alternative $x$ to another alternative $y$ in the same tournament, any tournament solution must either choose both $x$ and $y$ or neither of them.
The set of winners of a tournament $T$ with respect to a tournament solution $S$ is denoted by $S(T)$.
The tournament solutions considered in this paper are as follows:

\begin{itemize}
\item The \emph{Copeland set} (\cp) is the set of alternatives with the largest outdegree. 
The outdegree of an alternative is also referred to as its \emph{Copeland score}. 
\item The \emph{top cycle} (\tc) is the (unique) nonempty smallest set $X$ of alternatives such that $X \succ V(T) \setminus X$.
Equivalently, \tc is the set of alternatives that can reach every other alternative via a directed path.
\item The \emph{uncovered set} (\uc), is the set of alternatives that are not ``covered'' by any other alternative. 
An alternative $x$ is said to \emph{cover} another alternative $y$ if $D(y) \subseteq D(x)$. 
Equivalently, \uc is the set of alternatives reaching every other alternative via a directed path of length at most two.
\item The set of \emph{$k$-kings}, for an integer $k\geq 3$, is the set of alternatives that can reach every other alternative via a directed path of length at most $k$.
\item The \emph{Banks set} (\ba) is the set of alternatives that appear as the Condorcet winner of some transitive subtournament that cannot be extended.\footnote{
We say that an alternative $x \in V(T) \setminus V(T’)$ \emph{extends} a transitive subtournament $T’$ if $x$ dominates all alternatives in $T’$.
}
\end{itemize}

All of these tournament solutions satisfy \emph{Condorcet-consistency}, meaning that whenever a Condorcet winner exists, it is chosen as the unique winner.

It is clear from the definitions that \uc (the set of ``$2$-kings'') is contained in the set of $k$-kings for any $k\geq 3$, which is in turn a subset of \tc (the set of ``$(n-1)$-kings'', as any directed path has length at most $n-1$).
Moreover, both \cp and \ba are contained in \uc 
\citep{Laslier97}.

Given a tournament $T$ and an edge $e=(x,y)\in E(T)$, we let $\overline{e} := (y,x)$ denote its reversal.
Denote by $T^e$ the tournament that results from $T$ when reversing $e$.
A tournament solution $S$ is said to be \emph{monotonic} if for any edge $e=(y,x)\in E(T)$,
\[
x\in S(T) \quad \text{implies} \quad x\in S(T^e).
\]
In other words, a tournament solution is monotonic if a winner remains in the choice
set whenever its dominion is enlarged (while everything else is unchanged).
Equivalently, monotonicity means that a non-winner remains outside of the choice set whenever it becomes dominated by an additional alternative.

\subsection{Margin of Victory}

For a set of edges $R \subseteq E(T)$ of a tournament $T$, we define 
$\overline{R} := \{\overline{e}: e \in R\}$.
Denote by $T^R$ the tournament that results from $T$ when reversing all edges in~$R$, i.e.,
$V(T^R) = V(T)$ and $E(T^R) = (E(T) \setminus R) \cup \overline{R}$.

Fix a tournament solution $S$ and consider a tournament~$T$. An edge set $R \subseteq E(T)$ is called a \textit{destructive reversal set (\drs)} for $x \in S(T)$ if $x \notin S(T^R)$.
Analogously, $R$ is called a \textit{constructive reversal set (\crs)} for $x \in V(T)\setminus S(T)$ if $x \in S(T^R)$.
The \emph{margin of victory} of  $x \in S(T)$ is given by
\[\mov_S(x,T) = \min \{|R| \midd R \text{ is a \drs for $x$ in $T$} \},\] and for $x \not\in S(T)$ it is given by \[\mov_S(x,T)=-\min \{|R| \midd R \text{ is a \crs for $x$ in $T$} \}.\]

\newcommand{\movv}{\mov{} value\xspace}
\newcommand{\movvs}{\mov{} values\xspace}
By definition, $\mov_S(x,T)$ is a positive integer if $x \in S(T)$, and a negative integer otherwise.\footnote{The only exception is the degenerate case where $S$ selects all alternatives for all tournaments of some size $n$; in this case we define $\mov_S(x,T) = \infty$ for all alternatives $x$ and all tournaments~$T$ of that size. 
For ease of exposition, we will assume for the rest of the paper that the degenerate case does not occur, but all of our results still hold even when this case occurs.}

It follows from the definition of \mov{} that edge reversals have limited effects on the \movv of alternatives: If a single edge $e$ of a tournament $T$ is reversed, then $\mov_S(x,T)$ and $\mov_S(x,T^e)$ differ by at most $1$, unless $x$ is a winner in exactly one of the two tournaments $T$ and $T^e$ (in which case $|\mov_S(x,T)-\mov_S(x,T^e)|=2$).

Furthermore, \mov{} values behave monotonically with respect to edge reversals, provided the underlying tournament solution is monotonic. 

\begin{restatable}{proposition}{propmon}
\label{prop:mon}
Let $S$ be a monotonic tournament solution and consider two tournaments $T$ and $T^e$, where $e=(y,x) \in E(T)$. Then, $\mov_S(x,T^e) \ge  \mov_S(x,T)$. 
\end{restatable}

\begin{proof}
Let $T$ be a tournament and $e=(y,x) \in E(T)$, and let $T'=T^e$.
We first consider the case where $x \in S(T)$, \ie $\mov_S(x,T)>0$. By monotonicity of $S$, it follows that $x \in S(T')$. Suppose for contradiction that $\mov_S(x,T')<\mov_S(x,T)$, and let $R'$ be a minimal DRS for $x$ with respect to $T'$. 
Then, we will find a DRS $R$ for $x$ with respect to $T$ of size $|R|\le |R'|$, contradicting the assumption $\mov_S(x,T)>\mov_S(x,T')=|R'|$. To this end, define
$R = R' \setminus \{e\}$. To see that $R$ is a DRS for $x$ with respect to $T$, we need to show that $x \notin S(T^R)$. Since $R'$ is a DRS for $x$ with respect to $T'$, we have $x \notin S({T'}^{R'})$. In the case $e \in R'$, the tournament ${T'}^{R'}$ is identical to $T^R$, and the claim follows. 
In the case $e \notin R'$, we have ${T'}^{R'} = {T'}^R = T^{R \cup \{e\}}$.
Since $x \notin S(T^{R \cup \{e\}})$ and $S$ is monotonic, we have $x \notin S(T^R)$.

An analogous argument can be applied in the case where $x \in V(T) \setminus S(T)$, \ie $\mov_S(x,T)<0$. If $x \in S(T')$, then $\mov_S(x,T') > 0 >  \mov_S(x,T)$ holds trivially. Therefore, we assume that $x \in V(T) \setminus S(T')$, so $\mov_S(x,T') < 0$. 
Suppose for contradiction that $\mov_S(x,T')<\mov_S(x,T)$ and let $R$ be a minimal CRS for $x$ with respect to $T$. Then, by monotonicity, $R' = R \setminus \{e\}$ is a CRS for $x$ with respect to $T'$, contradicting the assumption that $\mov_S(x,T') < \mov_S(x,T) = -|R|$.
\end{proof}

\section{Structural Results}
\label{sec:structural}

In this section we provide a number of results relating the \mov{} notion to structural properties of the tournament in question. 
In particular, we identify conditions on tournament solutions ensuring that the corresponding \mov{} values are consistent with the covering relation (Section \ref{sec:cover-consistency}) and we examine the relationship between \mov{} values and Copeland scores (Sections \ref{sec:degree-consistency} and \ref{sec:probabilistic}). 
Our results are summarized in Table~\ref{tab:results}.

\subsection{Cover-Consistency}
\label{sec:cover-consistency}

Recall from \secref{sec:prelims} that an alternative $x$ covers another alternative $y$ if $D(y) \subseteq D(x)$. In particular, this implies that $x$ dominates $y$ (as otherwise $x \in D(y)$).  
The covering relation, which forms the basis for defining the uncovered set $\uc$, is transitive and has a close connection to Pareto dominance in voting settings \citep{BGH14a}. 

\begin{table}
\begin{center}
\begin{tabular}{ l c c c c c c c c c }
\toprule
& \pbox{2cm}{cover-\\cons.} & \pbox{2cm}{strong\\  deg.-cons.} & \pbox{2cm}{degree-\\cons.} & \pbox{2cm}{equal-\\  deg.-cons.}\\
\midrule
$\mov_\cp$ & \yes & \no & \yes & \no \\
$\mov_\tc$ & \yes & \yes & \yes & \yes \\
$\mov_\uc$ & \yes & \no & \no & \no \\
$\mov_{k\text{-kings}}$ & \yes & \no & \no & \no \\
$\mov_\ba$ & \yes & \no & \no & \no\\
\bottomrule
\end{tabular}
\caption{Consistency properties of the margin of victory for the tournament solutions $\cp$, $\tc$, $\uc$, $k$-kings, and $\ba$.}
\label{tab:results}
\end{center}
\end{table}

Intuitively, if $x$ covers $y$, there is a strong argument that $x$ is preferable to $y$. 
We show that for all of the tournament solutions that we consider, their corresponding \movvs are indeed consistent with this intuition. 

\begin{definition}
For a tournament solution $S$, we say that $\mov_S{}$ is \emph{cover-consistent} if, for any tournament $T$ and any alternatives $x,y\in V(T)$, $x$~covers~$y$ implies $\mov_S(x,T)\geq \mov_S(y,T)$.
\end{definition}

We introduce a new property that will be useful for showing that a tournament solution is cover-consistent.

\begin{definition}
A tournament solution $S$ is said to be \emph{transfer-monotonic} if for any edges $(y,z),(z,x)\in E(T)$,
\[
x\in S(T) \quad \text{implies} \quad x\in S(T'),
\]
where $T'$ is the tournament obtained from $T$ by reversing edges $(y,z)$ and $(z,x)$.
\end{definition}
\noindent In other words, if an alternative $x$ is chosen, then it remains chosen when an alternative $z$ is ``transferred'' from the dominion $D(y)$ of another alternative $y$ to its dominion $D(x)$.

We show that monotonicity and transfer-monotonicity together imply cover-consistency of the margin of victory.

\begin{lemma}
\label{lem:cover-consistency}
If a tournament solution $S$ is monotonic and transfer-monotonic, then $\mov_S$ satisfies cover-consistency.
\end{lemma}

\begin{proof}
Let $S$ be a monotonic and transfer-monotonic tournament solution, and suppose that alternative $x$ covers another alternative $y$ in a tournament $T$.
We will show that $\mov_S(x,T)\ge\mov_S(y,T)$.

If $x\in S(T)$ and $y\not\in S(T)$, the statement holds trivially since $\mov_S(x,T) > 0 > \mov_S(y,T)$.
Suppose for contradiction that $x\not\in S(T)$ and $y\in S(T)$.
Consider the tournament $T'$ obtained from $T$ by reversing the edge $(x,y)$ as well as edges $(x,z), (z,y)$ for each $z\in D(x) \setminus (D(y)\cup \{y\})$.
By monotonicity and transfer-monotonicity, $y\in S(T')$.
However, tournaments $T$ and $T'$ are isomorphic, and there is an isomorphism that maps $x\in T$ to $y\in T'$.
Since $x\not\in S(T)$, we must have $y\not\in S(T')$, a contradiction.

The remaining two cases are $x,y\in S(T)$ and $x,y\not\in S(T)$; both can be handled in an analogous manner, so let us focus on the latter case.
It suffices to show that given any CRS for $y$ of minimum size, we can construct a CRS of smaller or equal size for $x$.
Let $R_y$ be a CRS for $y$ of minimum size; we will construct a CRS $R_x$ for $x$ such that $|R_x|\le |R_y|$.

Let $A=V(T)\setminus\{x,y\}$, and partition $A$ into three sets $A_1=D(y)$, $A_2 = D(x)\setminus(D(y)\cup\{y\})$, and $A_3 = \overline{D}(x)$; see \Cref{fig:cover-consistency} for an illustration.
For any edge in $R_y$ between two alternatives of $A$, we add the same edge to $R_x$.
We do not add the edge $(x,y)$ regardless of whether it is present in~$R_y$.
Each remaining edge in $R_y$ is between an alternative in $A$ and one of $x,y$.
Note that $(y,a)\not\in R_y$ for any $a\in A$---otherwise, by monotonicity, removing such an edge would keep $R_y$ a CRS for $y$, contradicting the minimality of $R_y$.

\begin{figure}[!t]
\centering
\scalebox{0.8}{
\begin{tikzpicture}
\draw (6,.9) circle [radius = 0.6];
\draw (6,-.9) circle [radius = 0.6];
\node at (6,.9) {\Large$A_{2}$};
\node at (6,-.9) {\Large$A_{1}$};
\node[circle,fill,inner sep=2pt] at (3,.9)(x){}; 
\node[above=2pt] at (x){\Large$x$};
\node[circle,fill,inner sep=2pt] at (3,-.9)(y){}; 
\node[below=2pt] at (y){\Large$y$};
\draw (0,0) circle [radius = 0.6];
\node at (0,0) {\Large$A_{3}$};

\draw[->, thick] (x) -- (y);
\draw[->, thick] (x) -- (5.4,.9);
\draw[->, thick] (y) -- (5.4,-.9);
\draw[->, thick] (x) -- (5.4,-.65);
\draw[->, thick] (5.4,.65) -- (y);
\draw[->, thick] (0.55,.35) -- (x);
\draw[->, thick] (0.55,-.35) -- (y);

\end{tikzpicture}}

\caption{Illustration of the proof of Lemma~\ref{lem:cover-consistency}.}
\label{fig:cover-consistency}
\end{figure}
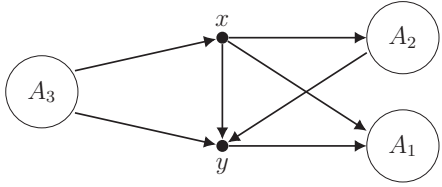

\smallskip \noindent
For each $a\in A$, we add further edges to $R_x$ as follows.
\vspace{-.5em}
\begin{itemize}
\item For $a\in A_1$:
\begin{itemize}
    \item If $(x,a)\in R_y$, add $(y,a)$ to $R_x$.
\end{itemize}
\item For $a\in A_2$:
\begin{itemize}
    \item If $(x,a)\in R_y$ but $(a,y)\not\in R_y$, add $(x,a)$ to $R_x$.
    \item If $(x,a)\not\in R_y$ but $(a,y)\in R_y$, add $(a,y)$ to $R_x$.
\end{itemize}
\item For $a\in A_3$:
\begin{itemize}
    \item If $(a,x)\in R_y$, add $(a,y)$ to $R_x$.
    \item If $(a,y)\in R_y$, add $(a,x)$ to $R_x$.
\end{itemize}
\end{itemize}
Clearly, $|R_x|\leq |R_y|$, and we have $y\in S(T^{R_y})$ by definition of $R_y$.
From $T^{R_y}$, we reverse the edge $(x,y)$ if it is present, and for $a\in A_2$ such that both $(x,a),(a,y)\not\in R_y$, we reverse $(x,a)$ and $(a,y)$.
Let $T'$ be the resulting tournament.
By monotonicity and transfer-monotonicity, we have $y\in S(T')$.
However, one can verify that there exists an isomorphism from $T'$ to $T^{R_x}$ that maps $x$ to $y$, $y$ to $x$, and every other alternative $a$ to itself.
Since $y\in S(T')$, we must have $x\in S(T^{R_x})$, meaning that $R_x$ is indeed a CRS for $x$.
\end{proof}

In \Cref{app:cover-consistency}, we show that neither monotonicity nor transfer-monotonicity can be dropped from the condition of Lemma~\ref{lem:cover-consistency}.
This also means that neither of the two properties implies the other.



We now show that all tournament solutions we consider in this paper satisfy both monotonicity and transfer monotonicity, thereby implying that their \mov{} functions are cover-consistent.

\begin{proposition}
\label{prop:monotonic-proofs}
\cp, \uc, \tc, $k$-kings, and \ba satisfy monotonicity.
\end{proposition}

\begin{proof}
It is already known that \cp, \uc, \tc, and \ba are monotonic \citep{Laslier97,BrandtBrHa16}; hence, it remains to establish the monotonicity of $k$-kings.
Let $x$ be a $k$-king in tournament $T$, and suppose that $T'$ is the tournament obtained by reversing an edge $(y,x)$.
Since any path of length at most $k$ from $x$ to another alternative in $T$ cannot contain the edge $(y,x)$, the same path is also present in $T'$.
Hence $x$ is also a $k$-king in $T'$.
\end{proof}

\begin{proposition}
\label{prop:transfer-proofs}
\cp, \uc, \tc, $k$-kings, and \ba satisfy transfer-monotonicity.
\end{proposition}

\begin{proof}
We start with \cp. 
If $x\in\cp(T)$ and edges $(y,z)$ and $(z,x)$ are reversed, then the outdegree of $x$ increases by~$1$, that of $y$ decreases by~$1$, while all other alternatives have the same outdegree as before.
Hence $x$ is in the Copeland set of the new tournament.

Next, we turn to $k$-kings. 
Let $x$ be a $k$-king in tournament~$T$, and suppose that $T'$ is the tournament obtained by reversing edges $(y,z)$ and $(z,x)$.
Consider a path of length at most $k$ from $x$ to another alternative $w$ in $T$; this path cannot contain the edge $(z,x)$.
If the path does not contain the edge $(y,z)$, then the same path also exists in $T'$.
Else, the path has the form $x\rightarrow \dots\rightarrow y\rightarrow z\rightarrow\dots\rightarrow w$, where possibly $z=w$.
We may then shorten this path to $x\rightarrow z\rightarrow\dots\rightarrow w$ in $T'$, meaning that $x$ can also reach $w$ in at most $k$ steps in~$T'$.
The proof for \uc and \tc proceeds in a similar manner.

Finally, let $x\in\ba(T)$, and consider an inclusion-maximal transitive subtournament with $x$ as the Condorcet winner.
Define $T'$ as in the previous paragraph.
Since $(z,x)\in E(T)$, $z$ does not belong to the subtournament, so all edges of the subtournament are intact in $T'$.
Since the subtournament is inclusion-maximal in $T$, no alternative different from $z$ extends it in $T'$.
Moreover, since $(x,z)\in E(T')$, $z$ cannot extend the subtournament in $T'$ either.
It follows that the subtournament is inclusion-maximal in $T'$, and therefore $x\in\ba(T')$.
\end{proof}

Lemma~\ref{lem:cover-consistency} and Propositions~\ref{prop:monotonic-proofs} and \ref{prop:transfer-proofs} together imply the following:

\begin{theorem}
\label{thm:cover-consistency-satisfies}
For each $S\in\{\cp,\tc,\uc,k\text{-kings},\ba\}$, $\mov_S$ satisfies cover-consistency.
\end{theorem}

In light of \Cref{thm:cover-consistency-satisfies}, one may wonder whether a stronger property, in which $x$ covers $y$ implies the strict inequality $\mov_S(x) > \mov_S(y)$, can also be achieved.
However, the answer is negative for all Condorcet-consistent tournament solutions, including all solutions that we consider.
Indeed, in a transitive tournament $x\succ y\succ z$ of size $3$, such a solution only selects $x$.
But since all three alternatives are chosen when they form a cycle (due to symmetry), both $y$ and $z$ can be brought into the winner set by reversing only one edge, so $\mov_S(y) = -1 = \mov_S(z)$ even though $y$ covers $z$.

\subsection{Degree-Consistency}
\label{sec:degree-consistency}

Given a tournament solution $S$ and a tournament $T$, the $\mov_S{}$ values yield a natural \emph{ranking} (possibly including ties) of the alternatives in $T$, where alternative $x$ is ranked higher than $y$ whenever $\mov_S(x,T)>\mov_S(y,T)$. 
We are interested in how closely this ranking by \movvs resembles the ranking by Copeland scores, according to which $x$ is ranked higher than $y$ if $\outdeg(x)> \outdeg(y)$. 

\begin{definition}
For a tournament solution $S$, we say that $\mov_S{}$ is
\begin{itemize}
\item \emph{degree-consistent} if, for any tournament $T$ and any alternatives $x,y\in V(T)$, $\outdeg(x)> \outdeg(y)$ implies $\mov_S(x,T)\geq \mov_S(y,T)$;
\item \emph{equal-degree-consistent} if, for any tournament $T$ and any alternatives $x,y\in V(T)$, $\outdeg(x)= \outdeg(y)$ implies $\mov_S(x,T)= \mov_S(y,T)$; and
\item \emph{strong degree-consistent} if, for any tournament $T$ and any alternatives $x,y\in V(T)$, $\outdeg(x)\geq \outdeg(y)$ implies $\mov_S(x,T)\geq \mov_S(y,T)$.
\end{itemize}
\end{definition}

It follows from the definitions that $\mov_S{}$ is strong degree-consistent if and only if it is both degree-consistent and equal-degree-consistent.
Observe also that cover-consistency is implied by degree-consistency.

We remark that these properties are not necessarily desirable from a normative perspective: 
Whereas the ranking implied by a strongly degree-consistent \mov{} function merely represents a coarsening of the straightforward ranking by outdegree, we are often interested in tournament solutions that take more structure of the tournament into account and, as a consequence, have $\mov$ functions that may violate (equal-)degree-consistency. 
Indeed, since degree-consistent \mov{} functions are in line with Copeland scores, their significance is somewhat limited and there would be little additional value derived from the \mov{} computations, which in some cases are much more involved than simply calculating Copeland scores.

We start by assessing the degree-consistency of $\mov_{\cp}$ and show that it satisfies degree-consistency but not equal-degree-consistency.

\begin{proposition}
$\mov_{\cp}$ does not satisfy equal-degree-consistency.  \label{thm:co-equal-consistency}
\end{proposition}

\begin{proof} 
We construct a counterexample with seven alternatives $x,z,y_1,y_2,y_3,y_4$ and $y_5$; see Figure \ref{fig:co-equal-consistency} for an illustration. 
Alternative $z$ is the unique Copeland winner with an outdegree of $5$, alternatives $x,y_3,y_4$ and $y_5$ have outdegree $3$, and $y_1$ and $y_2$ have outdegree $2$. 
We argue that, even though $y_3$ and $x$ have the same outdegree, it holds that $\mov_{\cp}(y_3,T)=-1$ and $\mov_{\cp}(x,T)=-2$. 
The former holds since $y_3$ can be made a Copeland winner by reversing the edge $(y_3,z)$. 
For $x$, however, there does not exist an edge whose reversal simultaneously strengthens $x$ and weakens $z$. 
Hence, we need to reverse at least two edges, e.g., $(x,y_3)$ and $(x,y_4)$, in order to make $x$ a Copeland winner.  
\begin{figure}[!h]
\centering
\scalebox{0.8}{
\begin{tikzpicture}

\node [circle, very thick, draw=black, fill=black!10, inner sep=1.4cm](d) at (0,0){}; 
\node [circle, inner sep=2pt, fill=black](x) at (-3.5,0){};\node [below = 2pt] at (x){\Large$x$};
\node [circle, inner sep=2pt, fill=black](z) at (3.5,0){};\node [below = 2pt] at (z){\Large$z$};
\def \n {5}
\def \m {4}
\def \radius {1.2cm}
\def \margin {8} 
\foreach \s in {1,2,3,4,5}
{ 
      \node[circle,fill,inner sep=2pt](\s) at ({360/5 * (\s - 1)}:\radius) {}; 
      }

\node[above right] at (1){\Large$y_4$};
\node[above right] at (2){\Large$y_3$};
\node[above left] at (3){\Large$y_2$};
\node[below left] at (4){\Large$y_1$};
\node[below right] at (5){\Large$y_5$};

\draw[->,thick] (x) -- (4);
\draw[->,thick] (x) -- (3);
\draw[->,thick] (x) to[bend left=80](z);

\foreach \s/\t in {1/2,1/3,2/3,2/4,3/4,3/5,4/5,4/1,5/1,5/2}{
	\draw[->,thick] (\t) -- (\s);}
\end{tikzpicture}
}
\caption{Illustration of the example in the proof of Proposition~\ref{thm:co-equal-consistency}. Missing edges point from right to left.}
\label{fig:co-equal-consistency}
\end{figure}
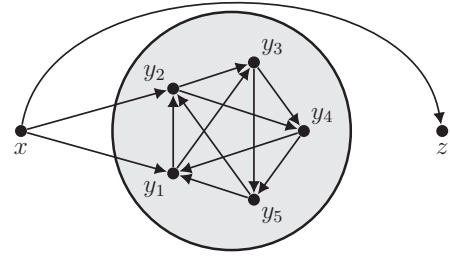
\end{proof}

\begin{proposition}
$\mov_{\cp}$ satisfies degree-consistency.  
\end{proposition}

\begin{proof}
Let $T$ be a tournament, $x\in \cp(T)$, $y\in V(T)\setminus \cp(T)$, and $\delta = \outdeg(x) - \outdeg(y)$.
We claim that 
$
-\delta\leq \mov_{\cp}(y,T) \leq -(\delta-1)
$.
The left inequality follows from the fact that if we reverse $\delta$ incoming edges into $y$, then $y$ becomes a Copeland winner.
For the right inequality, note that each edge reversal decreases the difference $\outdeg(x) - \outdeg(y)$ by at most $1$; the only exception is the edge $(x,y)$, in which case the difference decreases by $2$.
Since the difference starts at $\delta$ and must be nonpositive in order for $y$ to become a Copeland winner, at least $\delta-1$ edges must be reversed.

Now, let $v,w$ be arbitrary alternatives in $T$ such that $\outdeg(v) > \outdeg(w)$. 
We have $w\not\in \cp(T)$.
If $v\in\cp(T)$, then $\mov_{\cp}(v,T) > 0 > \mov_{\cp}(w,T)$.
Assume that $v\not\in\cp(T)$.
Considering an alternative $u\in\cp(T)$, we have  \begin{align*}
\mov_{\cp}(v,T) &- \mov_{\cp}(w,T)\\ 
&\geq -(\outdeg(u) - \outdeg(v)) \\
& \qquad + (\outdeg(u) - \outdeg(w) - 1) \\ 
&= \outdeg(v) - \outdeg(w) -1 \\
&\geq 0,
\end{align*} 
meaning that $\mov_{\cp}$ is degree-consistent.
\end{proof}

Next, we consider the top cycle. Recall that, for a given tournament $T$ of size $n$,  $\tc$ coincides with $k$-kings for $k=n-1$. In order to show that $\mov_{\tc}$ satisfies strong degree-consistency, we need two lemmas (one of which is already known). We first introduce some notation. 

Given a tournament $T$ and distinct alternatives $x,y\in V(T)$, an edge set $R\subseteq E(T)$ is said to be a \emph{$k$-length bounded $x$-$y$-cut} if, once $R$ is removed, every path from $x$ to $y$ has length strictly  greater than $k$.
Denote by $\mincut_k(x,y)$ the size of a smallest $k$-length bounded $x$-$y$-cut.
A set $R$ is said to be a \emph{$k$-length bounded $x$-cut} if it is a $k$-length bounded $x$-$y$-cut for some $y\in T$.

\begin{lemma}[Lemma~4 by \citemov]
\label{lem:xcut-kkings-destr}
For any $k\in\{2,3,\dots,n-1\}$, a set $R \subseteq E(T)$ is a minimum DRS for $x$ with respect to $k$-kings if and only if $R$ is a minimum $k$-length bounded $x$-cut in $T$.
\end{lemma}

The next lemma establishes a surprisingly succinct relation between the sizes of the minimum cuts with respect to a pair of alternatives, and can be shown using the max-flow min-cut theorem.\footnote{The lemma also follows from a more general statement by \citet[Prop.~16]{BubboloniGo18}. We thank Daniela Bubboloni and Michele Gori for pointing this out to us.}
Define $\mincut(x,y) = \mincut_{n-1}(x,y)$.

\begin{restatable}{lemma}{lemtournamentssizexycuts}\label{lem:tournaments-size-xycuts}
Let $T$ be a tournament and $x,y \in V(T)$. Then, \[\mincut(x,y) - \mincut(y,x) = \outdeg(x)-\outdeg(y).\]
\end{restatable}

\begin{proof}

For ease of presentation, we divide the alternatives in $V(T)\backslash\{x,y\}$ into four sets: 
\begin{itemize}
\item $D_x$ consists of the alternatives dominated by $x$ but not $y$;
\item $D_y$ consists of the alternatives dominated by $y$ but not $x$;
\item $D_{xy}$ consists of the alternatives dominated by both $x$ and~$y$;
\item $D_0$ consists of the alternatives dominated by neither $x$ nor~$y$.
\end{itemize}
See Figure \ref{fig:tournaments-mincuts} for an illustration. 
\begin{figure}[!t]
\centering
\scalebox{0.8}{
\begin{tikzpicture}
\draw (0,0) circle [radius = 0.7];
\draw (-2,0) circle [radius = 0.7];
\draw (2,0) circle [radius = 0.7];
\node at (0,0) {\Large$D_{xy}$};
\node at (-2,0) {\Large$D_{x}$};
\node at (2,0) {\Large$D_{y}$};
\node[circle,fill,inner sep=2pt,above] at (-1,2)(x){}; 
\node[left=2pt] at (x){\Large$x$};
\node[circle,fill,inner sep=2pt,above] at (1,2)(y){}; 
\node[right=2pt] at (y){\Large$y$};
\draw (0,4.2) circle [radius = 0.7];
\node at (0,4.2) {\Large$D_{0}$};
\draw[->, thick] (-0.45,3.65) -- (x);
\draw[->, thick] (0.45,3.65) -- (y);
\draw[->, thick] (x) -- (-0.3,0.7);
\draw[->, thick] (y) -- (0.3,0.7);
\draw[->, thick] (y) -- (1.7,0.7);
\draw[->, thick] (1.5,0.7) -- (-0.8,1.9);
\draw[->, ultra thick,red!70!black,dashed] (-1.5,0.7) -- (0.8,1.9);
\draw[->,ultra thick,red!70!black,dashed] (x) -- (-1.7,0.7);
\end{tikzpicture}
}

\caption{Illustration of the proof of \Cref{lem:tournaments-size-xycuts}.}
\label{fig:tournaments-mincuts}
\end{figure}

Call a path from $x$ to $y$ an \emph{$x$-$y$-path}.
From the max-flow min-cut theorem \citep{FoFu56a}, the size of a minimum cut from $x$ to $y$ equals the maximum number of edge-disjoint $x$-$y$-paths (and analogously for a minimum cut from $y$ to $x$). 
Making use of this fact, we will argue about maximum sets of edge-disjoint paths instead of minimum cuts.
Let $\mathcal{P}_x$ be the set of all paths of length one or two from $x$ to $y$. 
Similarly, let $\mathcal{Q}_y$ be the set of all paths of length one or two from $y$ to $x$.
\begin{claim}
There exists a maximum set of edge-disjoint $x$-$y$ paths, $\mathcal{P}$, such that $\mathcal{P}_x \subseteq \mathcal{P}$, and a maximum set of edge-disjoint $y$-$x$ paths, $\mathcal{Q}$, such that $\mathcal{Q}_y \subseteq \mathcal{Q}$. 
\end{claim}

\begin{proof}[Proof of Claim]
By symmetry, it suffices to prove the former statement.
Let $\mathcal{P}$ be a maximum set of edge-disjoint $x$-$y$ paths.
We show how we can alter $\mathcal{P}$ in an iterative manner so that $\mathcal{P}_x \subseteq \mathcal{P}$ holds while $\mathcal{P}$ remains a maximum set of edge-disjoint $x$-$y$ paths. 

If $(x,y) \in E(T)$, then also $\{x\rightarrow y\} \in \mathcal{P}$, since otherwise $\mathcal{P}$ cannot be maximum. 
Next, consider some $z \in D_x$, and let $F = \{(x,z),(z,y)\}$.
If there exists exactly one path in $\mathcal{P}$ which contains an edge from $F$, we replace this path by the path $x\rightarrow z\rightarrow y$. 
Else, if there exist two paths $P_1$ and $P_2$ which contain an edge from $F$, then we can assume without loss of generality that $P_1$ starts with the edge $(x,z)$ and $P_2$ ends with the edge $(z,y)$. 
In this case, we replace $P_1$ by $x\rightarrow z\rightarrow y$, and construct $P_2$ by joining the remaining parts of the two paths to go from $x$ to $y$ through $z$, possibly omitting any cycles that arise. 
The newly created paths are edge-disjoint with respect to all other paths in $\mathcal{P}$, and the number of paths in $\mathcal{P}$ remains unchanged. 
At the end of this process, we have $\mathcal{P}_x \subseteq \mathcal{P}$.
\end{proof}

Using the Claim, we let $\mathcal{P}$ (resp., $\mathcal{Q}$) be a maximum set of edge-disjoint $x$-$y$-paths (resp., $y$-$x$-paths) such that $\mathcal{P}_x \subseteq \mathcal{P}$ (resp., $\mathcal{Q}_y \subseteq \mathcal{Q}$) holds. 
We next show that $|\mathcal{P} \setminus \mathcal{P}_x| = |\mathcal{Q} \setminus \mathcal{Q}_y|$. 
Suppose that this is not the case, and assume without loss of generality that $|\mathcal{P} \setminus \mathcal{P}_x| > |\mathcal{Q} \setminus \mathcal{Q}_y|$. 
From $\mathcal{P} \setminus \mathcal{P}_x$, we will construct a set of edge-disjoint $y$-$x$ paths, $\mathcal{Q}'$, which is also edge-disjoint to all paths in $\mathcal{Q}_y$ and is of size $|\mathcal{Q}'|=|\mathcal{P} \setminus \mathcal{P}_x|$, so that $\mathcal{Q}_y\cup \mathcal{Q}'$ contradicts the maximality of $\mathcal{Q}$. 

To this end, let $P \in \mathcal{P}\setminus \mathcal{P}_x$. Note that $P$ is of the form $x\rightarrow v_1\rightarrow\dots\rightarrow v_\ell \rightarrow y$ for some $2\leq \ell\leq n-1$. 
Also, $v_1 \in D_{xy}$, since otherwise $P$ would intersect with a path in $\mathcal{P}_x$. 
For the same reason, $v_\ell \in D_0$. 
Hence, $(y,v_1),(v_\ell,x) \in E(T)$ and therefore $y\rightarrow v_1\rightarrow\dots\rightarrow v_\ell \rightarrow x$, where the part between $v_1$ and $v_\ell$ is the same as in $P$, is a $y$-$x$-path in $T$. 
We create $\mathcal{Q}'$ by using this mirroring argument for all paths in $\mathcal{P}\setminus\mathcal{P}_x$. 
By construction, the paths in $\mathcal{Q}'$ are edge-disjoint with respect to the paths in $\mathcal{Q}_y$, so $\mathcal{Q}'$ has the desired property.
Hence, $|\mathcal{P} \setminus \mathcal{P}_x| = |\mathcal{Q} \setminus \mathcal{Q}_y|$.

Finally, we have 
\begin{align*}
\mincut(x,y) &- \mincut(y,x) \\
&= |\mathcal{P}| - |\mathcal{Q}| \\
&= |\mathcal{P}_x| + |\mathcal{P}\setminus \mathcal{P}_x|  - |\mathcal{Q}\setminus \mathcal{Q}_y| - |\mathcal{Q}_y| \\
&= |\mathcal{P}_x| - |\mathcal{Q}_y|\\
&= \outdeg(x) - \outdeg(y),
\end{align*}
as desired.
\end{proof}

\begin{theorem}\label{thm:tc-degree}
$\mov_{\tc}$ satisfies strong degree-consistency. 
\end{theorem}

\begin{proof}
Fix a tournament $T$ and let $x,y \in V(T)$ with $\outdeg(x) \geq \outdeg(y)$. First, we show that $x,y \in \tc(T)$ constitutes the only non-trivial case. 
Since all alternatives in $\tc(T)$ dominate all alternatives outside, it cannot be that $x\not\in \tc(T)$ and $y\in\tc(T)$.
If $x,y \not\in \tc(T)$, \citemov showed that $\mov_{\tc}(x,T)=-1=\mov_{\tc}(y,T)$.
If $x\in \tc(T)$ and $y\not\in \tc(T)$, then $\mov_{\tc}(x) > 0 > \mov_{\tc}(y)$.

Assume now that $x,y \in \tc(T)$. Let $R$ be a minimum DRS for $x$. 
By Lemma \ref{lem:xcut-kkings-destr} with $k=n-1$, we know that $R$ is a minimum $x$-$t$-cut for some $t \in V(T)$. 
We consider two cases. 
First, assume that $R$ is also a $y$-$t$-cut. 
Then, a minimum $y$-$t$-cut $R'\subseteq E(T)$ satisfies $|R'| \leq |R|$, proving that $\mov_{\tc}(x,T)=|R| \geq |R'|\geq \mov_{\tc}(y,T)$. 
For the second case, assume that $R$ is not a $y$-$t$-cut. 
Then, $R$ needs to be an $x$-$y$-cut (since otherwise $x$ can reach~$t$ via~$y$), and therefore it must be a minimum $x$-$y$-cut.
By Lemma~\ref{lem:tournaments-size-xycuts}, since $\outdeg(x)\ge\outdeg(y)$, for a minimum $y$-$x$-cut $R'$ it holds that $|R|\geq|R'|$.
Hence $\mov_{\tc}(x,T) = |R| \geq |R'| \geq \mov_{\tc}(y,T)$.
\end{proof}
On the other hand, we show in the next three propositions that $\uc$, $\ba$, and $k$-kings do not satisfy any of the degree-consistency properties. 

\begin{proposition}
$\mov_{\uc}$ and $\mov_{\ba}$ do not satisfy equal-degree-consistency.  \label{thm:uc-equal-consistency}
\end{proposition}

\begin{proof}
\begin{figure}[!t]
\centering
\scalebox{0.8}{
\begin{tikzpicture}

\node [circle, very thick, draw=black, fill=black!10, inner sep=1cm](d) at (0,0){}; 
\node [circle, very thick, draw=black, fill=black!10, inner sep=1cm](d2) at (5,0){};
\node [circle, inner sep=2pt, fill=black](g) at (2.5,-2){};\node [above = 2pt] at (g){\Large$g$};
\def \n {5}
\def \m {4}
\def \radius {.8cm}
\def \margin {8} 
\foreach \s in {1,...,3}
{ 
      \node[circle,fill,inner sep=2pt](\s) at ({360/3 * (\s - 1)}:\radius) {}; 
      }

\foreach \s in {7,...,9}
{ 
      \node[circle,fill,inner sep=2pt,xshift=5cm](\s) at ({360/3 * (\s - 1)+180}:\radius) {};
}

\node[above right] at (1){\Large$a$};
\node[above right] at (2){\Large$b$};
\node[below right] at (3){\Large$c$};
\node[above left] at (7){\Large$d$};
\node[below left] at (8){\Large$f$};
\node[above=.1cm] at (9){\Large$e$};

\foreach \s/\t in {2/1,1/3,3/2,
			7/8,9/7,8/9}{
	\draw[<-,thick] (\s) -- (\t);}
	
\draw[->, thick] (1) -- (7);
\draw[->, thick] (2) -- (9);
\draw[->, thick] (3) -- (8);

\draw[->, thick] (1,-1) to[bend right=10] (g);
\draw[<-, thick] (4,-1) to[bend left=10] (g);

\end{tikzpicture}
}

\caption{Illustration of the example in the proof of Proposition~\ref{thm:uc-equal-consistency}. Missing edges point from right to left.}
\label{fig:uc-equal-consistency}
\end{figure}
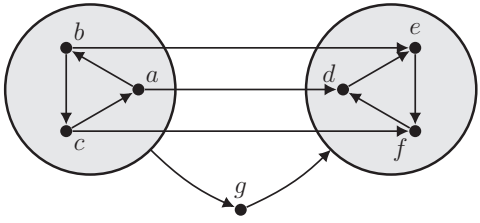
We give a counterexample for both $\mov_{\uc}$ and $\mov_{\ba}$ at once. 
The example tournament $T$ contains seven alternatives, $a,b,c,d,e,f,g$, which all have the same outdegree. 
See Figure \ref{fig:uc-equal-consistency} for an illustration.\footnote{This tournament has been previously considered by \citet[Fig.~7]{BrandtBrSe18}.} 

We start by showing that this is a counterexample for $\mov_{\uc}$. 
Note that all alternatives are in the uncovered set of this tournament. 
We claim that $\mov_{\uc}(g,T) = 2$ while $\mov_{\uc}(d,T) = 1$.
The former claim follows from \Cref{lem:xcut-kkings-destr} and the observation that $g$ has exactly two edge-disjoint paths of length at most two to every other alternative.
For the latter, note that $d\rightarrow b$ is the only path of length at most two from $d$ to $b$.

Next, we show that the counterexample holds for $\mov_{\ba}$ as well. To this end, we show that $g,d \in \ba(T)$ and  $\mov_{\ba}(g,T)>1$. Since $BA(T) \subseteq UC(T)$ and therefore $\mov_{\ba}(d,T)\leq \mov_{\uc}(d,T)=1$, this suffices to proof the claim. 
Since the transitive subtournament consisting of $g,d,e$ cannot be extended, $g \in \ba(T)$.
Likewise, $d\in \ba(T)$ because the subtournament consisting of $d,b,e$ cannot be extended.

It remains to argue that $\mov_{\ba}(g,T) > 1$. Assume for contradiction that $\mov_{\ba}(g,T) = 1$, i.e., there exists an edge whose reversal takes $g$ out of the Banks set.
Suppose first that $g$ still dominates all of $d,e,f$ after the reversal.
Then, since each of $a,b,c$ is dominated by two of $d,e,f$ in $T$, at least one of the three transitive subtournaments with alternative set $\{g,d,e\},\{g,e,f\},\{g,f,d\}$ cannot be extended.
The remaining case is that $g$ no longer dominates all of $d,e,f$, meaning that an edge $(g,x)$ is reversed for some $x\in\{d,e,f\}$.
Assume without loss of generality that $x=d$.
In this case, the transitive subtournament formed by $g,e,f$ still cannot be extended, so $\mov_{\ba}(g,T) > 1$.
This concludes the proof. 
\end{proof}

\begin{proposition}
$\mov_{\uc}$ and $\mov_{\ba}$ do not satisfy degree-consistency. \label{thm:uc-degree-consistency}
\end{proposition}

\begin{proof}
\begin{figure}[!t]
\centering
\scalebox{0.8}{
\begin{tikzpicture}

\node [circle, very thick, draw=black, fill=black!10, inner sep=1cm](d) at (0,3){}; 
\node [circle, very thick, draw=black, fill=black!10, inner sep=1cm](d2) at (0,-1){};
\node [circle, inner sep=2pt, fill=black](y4) at (0,1){};\node [right = 2pt] at (y4){\Large$y_4$};
\node [circle, inner sep=2pt, fill=black](x) at (-3.5,1){};\node [left = 2pt] at (x){\Large$x$};
\node [circle, inner sep=2pt, fill=black](z) at (3.5,1){};\node [right = 2pt] at (z){\Large$z$};
\def \n {5}
\def \m {4}
\def \radius {.8cm}
\def \margin {8} 
\foreach \s in {1,...,3}
{ 
      \node[circle,fill,inner sep=2pt,yshift=3cm](\s) at ({360/3 * (\s - 1)}:\radius) {}; 
      }

\foreach \s in {7,...,9}
{ 
      \node[circle,fill,inner sep=2pt,yshift=-1cm](\s) at ({360/3 * (\s - 1)}:\radius) {};
}

\node[above right] at (1){\Large$y_1$};
\node[above right] at (2){\Large$y_2$};
\node[below right] at (3){\Large$y_3$};
\node[above right] at (7){\Large$y_5$};
\node[above right] at (8){\Large$y_6$};
\node[below right] at (9){\Large$y_7$};

\foreach \s/\t in {2/1,3/1,3/2,
			8/7,7/9,9/8}{
	\draw[<-,thick] (\s) -- (\t);}
	
\draw[->, thick] (0,1.6) -- (0,1.2);
\draw[->, thick] (0,0.8) -- (0,0.5);
\draw[->, thick] (1,0) to[bend right] (1,1.8);
\draw[->,thick] (1) -- (z);
\draw[->, thick] (1.4,-1) -- (z);
\draw[->, thick] (x) -- (-1.5,3);
\draw[->, thick] (x) -- (-0.2,1);
\end{tikzpicture}
}
\caption{Illustration of the example in the proof of Proposition~\ref{thm:uc-degree-consistency}. Missing edges point from right to left.}
\label{fig:uc-degree-consistency}
\end{figure}
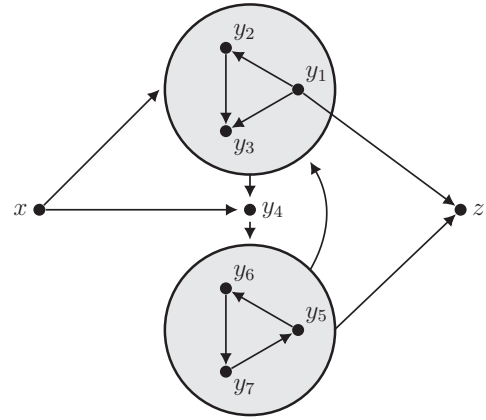
We give a counterexample for both $\mov_{\uc}$ and $\mov_{\ba}$ at once. 
The example tournament $T$ consists of nine alternatives, $x,z$, and $y_i$ for $i=1,\dots,7$. See Figure \ref{fig:uc-degree-consistency} for an illustration. 
Alternative $x$  dominates exactly $y_1,y_2,y_3$ and $y_4$, and alternative $z$ dominates $x,y_2,y_3$ and $y_4$. In general $y_i$ dominates $y_j$ whenever $i <j$, with the exceptions that $y_7$ dominates $y_5$, and $\{y_5,y_6,y_7\}$ dominate $\{y_1,y_2,y_3\}$.

We start by showing that this is a counterexample for $\mov_{\uc}$.
First, observe that both $x$ and $y_4$ belong to $\uc(T)$. 
Moreover, $x$ has outdegree $4$, and $\mov_{\uc}(x,T)=1$ by \Cref{lem:xcut-kkings-destr} since there is only one path of length at most two from $x$ to $z$. 
On the other hand, $y_4$ has outdegree $3$, but its $\mov_{\uc}$ is $2$. 
To see this, note that $y_4$ can reach each alternative in $\{y_5,y_6,y_7\}$ both directly and through another alternative in this latter set.
Moreover, since $y_5,y_6,y_7$ all dominate the remaining alternatives, $y_4$ has three disjoint paths of length two to each of these alternatives.

Next, we show that the counterexample holds for $\mov_{\ba}$ as well. 
To this end we show that $x,y_4 \in \ba(T)$ and  $\mov_{\ba}(y_4,T)>1$.
Since $BA(T) \subseteq UC(T)$ and therefore $\mov_{\ba}(x,T)\leq \mov_{\uc}(x,T)=1$, this suffices to establish the claim. 
In order to show that $y_4 \in \ba(T)$, we define $T_{56}$ to be the subtournament of $T$ induced by the set $\{y_4, y_5, y_6\}$. 
Analogously, we define $T_{67}$ and $T_{75}$. It is easy to see that $T_{56}, T_{67}$ and $T_{75}$ are all transitive and $y_4$ is their maximum element. 
Moreover, none of them can be extended by any other alternative, meaning that $y_4 \in \ba(T)$. 
In order to show that $x \in \ba(T)$, consider the subtournament induced by $\{x,y_1,y_2,y_3,y_4\}$, and observe that it cannot be extended by any other alternative. 

It remains to argue that $\mov_{\ba}(y_4,T) > 1$. 
Assume for contradiction that $\mov_{\ba}(y_4,T) = 1$ and let $\{(a,b)\}$ be a destructive reversal set for $y_4$, i.e., $y_4 \not\in \ba(T')$, where $T'$ is obtained from $T$ by reversing the edge $(a,b)$. 
We do a case distinction on the identity of $a$ and $b$. 
First, consider the cases where $a,b \in \{x,y_1,y_2,y_3,y_4,z\}$ or $a,b \in \{y_5,y_6,y_7\}$. 
Then,  $T'_{56}, T'_{67}$ and $T'_{75}$ (defined analogously as for $T$) are transitive subtournaments with maximal element $y_4$ which cannot be extended. 
Second, let one of $a$ and $b$ be from $\{x,y_1,y_2,y_3,y_4,z\}$ while the other one is from $\{y_5,y_6,y_7\}$, and without loss of generality let $\{a,b\} \cap \{y_5,y_6,y_7\} = \{y_5\}$. 
Then, the subtournament $T'_{67}$ is still transitive, has $y_4$ as a maximal element, and cannot be extended.
It follows that $y_4 \in \ba(T')$, a contradiction to the assumption that $\{(a,b)\}$ is a destructive reversal set. 
This concludes the proof. 
\end{proof}

\begin{figure*}[t]
\centering
\scalebox{0.8}{
\begin{tikzpicture}
\fill [black!20, rounded corners=2ex] (1.5,3) rectangle (2.5,-1);

\node[circle,fill,inner sep=2pt] at (0,2)(x){}; \node[above=2pt] at (x){\Large$x$};
\node[circle,fill,inner sep=2pt] at (0,0)(y){}; \node[below=2pt] at (y){\Large$y$};
\draw[->, thick] (y) -- (x); 

\node[circle,draw,very thick,inner sep=5pt] at (2,2)(a1){$\alpha$}; 
\node[circle,draw, very thick,inner sep=5pt] at (2,0)(b1){$\beta$}; 

\draw[->, thick] (x) -- (a1); 
\draw[->, thick] (y) -- (b1); 

\node at (3.5,2)(da){$\dots$}; 
\node at (3.5,0)(db){$\dots$}; 

\node[circle,draw,very thick,inner sep=5pt] at (5,2)(a2){\Large$\alpha$};
\node[circle,draw, very thick,inner sep=5pt] at (5,0)(b2){\Large$\beta$}; 

\draw[->, thick] (a1) -- (da);  \draw[->, thick] (da) -- (a2); 
\draw[->, thick] (b1) -- (db);  \draw[->, thick] (db) -- (b2); 
\draw[->,thick](a1) -- (b1);

\node[circle,fill,inner sep=2pt] at (7,0)(z){}; \node[below=2pt] at (z){\Large$z$};
\node[circle,draw, very thick,inner sep=5pt] at (7,2)(a3){\Large$\alpha$}; 

\draw[->,thick] (a2) -- (a3); 
\draw[->,thick] (b2) -- (z); 
\draw[->,thick] (a2) -- (b2);

\node[circle,fill,inner sep=2pt] at (9,1)(t){}; \node[above=2pt] at (t){\Large$t$};
\draw[->,thick] (a3) -- (z);
\draw[->,thick] (a3) -- (t); 
\draw[->,thick] (z) -- (t); 

\draw [decorate,thick, black!60, decoration={brace,amplitude=10pt, raise=20pt, mirror}] (b1.west) -- (b2.east);
\node at (3.5,-1.3)(d){\Large$\times (k-2)$};

\end{tikzpicture} \hspace{.5cm} \vline \hspace{1cm}\begin{tikzpicture}

\fill [black!20, rounded corners=2ex] (-1.3,4.5) rectangle (2,-1.5);
\node [circle, very thick, draw=black, fill=white, inner sep=.8cm](d) at (0,3){}; \node[above right=1cm] at (d){\Large$\alpha=5$};
\node [circle, very thick, draw=black, fill=white, inner sep=.8cm](d2) at (0,0){}; \node[below right=1cm] at (d2){\Large$\beta=4$};

\def \n {5}
\def \m {4}
\def \radius {.8cm}
\def \margin {8} 
\foreach \s in {1,...,5}
{ 
      \node[circle,fill,inner sep=2pt,yshift=3cm](\s) at ({360/\n * (\s - 1)}:\radius) {}; 
}

\foreach \s in {6,...,9}
{ 
      \node[circle,fill,inner sep=2pt,yshift=0cm](\s) at ({360/\m * (\s - 1)}:\radius) {};
}

\foreach \s/\t in {1/2,1/3,2/3,2/4,3/4,3/5,4/5,4/1,5/1,5/2,
			6/7,6/8,7/8,7/9,8/9,9/6}{
	\draw[<-,thick] (\s) -- (\t);}
	
\draw[->,ultra thick] (0,1.8) -- (0,1.2);
\end{tikzpicture}
}

\caption{Illustration of the example in the proof of Proposition~\ref{thm:kkings-degree-consistency}. Missing edges point from right to left. The left image gives an overview of the example, while the right image shows a close-up of two ``supernodes'' of size $\alpha=5$ and $\beta=4$, respectively.}
\label{fig:kkings-degree-consistency}
\end{figure*}
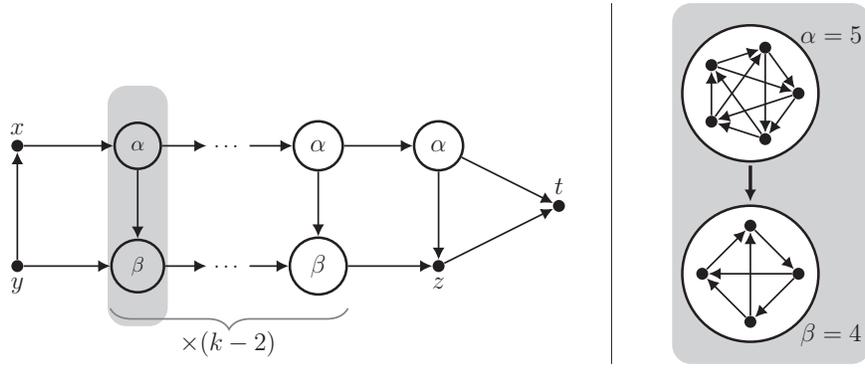

\begin{proposition}
$\mov_{k\emph{-kings}}$ (for constant $k\ge 3$) satisfies neither degree-consistency nor equal-degree-consistency. \label{thm:kkings-degree-consistency}
\end{proposition}

\begin{proof}
Let $k \ge 3$ be a constant. 
We describe a family of examples which, after specifying two parameters, allows us to disprove the degree-consistency as well as the equal-degree-consistency of $\mov_{k\text{-kings}}$.
The high-level idea of the instance is as follows: There exist two alternatives, $x$ and $y$, both of which are currently $k$-kings. Moreover, $\outdeg(x)=\alpha$ and $\outdeg(y) = \beta+1$, where $\alpha$ is any odd positive integer while $\beta$ can be any positive integer. 
The example is constructed in such a way that the $\mov_{k\text{-kings}}$ of $x$ is at least $(\alpha+1)/2$ while that of $y$ is $1$. 
Setting $\alpha\geq 3$ and $\beta=\alpha-1$ yields a violation of equal-degree-consistency, and setting $\alpha\geq3$ and $\beta\geq \alpha$ yields a violation of degree-consistency. 

We now describe the construction in more detail; see \Cref{fig:kkings-degree-consistency} for an illustration. 
The tournament $T$ consists of four singleton alternatives, $x,y,z$ and $t$, and $2k-3$ \emph{supernodes}.
These supernodes are tournaments themselves, where all alternatives in the supernode have the same relation to each alternative outside of the supernode. 
In \Cref{fig:kkings-degree-consistency}, we depict supernodes by large circles. 
In our construction there exist two different types of supernodes: those with parameter $\alpha$ and those with parameter $\beta$. 
Each supernode with parameter $\alpha$ contains $\alpha$ singleton alternatives and have a specific structure. 
More precisely, the alternatives are arranged on a cycle and each alternative dominates exactly the $(\alpha-1)/2$ alternatives following it on the cycle. (This structure is called a ``cyclone'' in \Cref{app:tc-example}.)
We make fewer specifications for supernodes of parameter $\beta$ and simply require that each of them corresponds to a tournament of size $\beta$, but their inner structure can be chosen arbitrarily. 
For the relationships between alternatives and supernodes, we refer to the left image of Figure \ref{fig:kkings-degree-consistency}.

\begin{claim}
Let $\alpha$ be an odd positive integer, $\beta$ be a positive integer, and $T$ be a tournament with parameters $\alpha$ and $\beta$ as described in Figure \ref{fig:kkings-degree-consistency}. 
Then, $\outdeg(x)=\alpha$, $\outdeg(y)=\beta+1$, $\mov_{k\text{-kings}}(x,T)\geq (\alpha+1)/2 $, and $\mov_{k\text{-kings}}(y,T)=1$. 
\end{claim}
 
 \begin{proof}[Proof of Claim]
 The outdegrees of $x$ and $y$ follow by construction, and it can be verified that both $x$ and $y$ are $k$-kings.
 We start by showing that the $\mov_{k\text{-kings}}$ of $y$ in the constructed example is $1$. 
 By Lemma \ref{lem:xcut-kkings-destr}, it suffices to show that there exists a $k$-length bounded $y$-$t$-cut of size $1$. 
 The edge $(z,t)$ forms such a cut: after deleting it, all paths from $y$ to $t$ have length at least $k+1$. 
 
Next, we show that the $\mov_{k\text{-kings}}$ of $x$ is at least $(\alpha+1)/2$. 
We do so by arguing that for any $w \in V(T)$, the size of a minimum $k$-length bounded $x$-$w$-cut is at least $(\alpha+1)/2$. 
To this end, we give a lower bound on the number of edge-disjoint paths of length at most $k$ from $x$ to $w$; clearly, any $k$-length bounded $x$-$w$-cut must have size at least this latter number. 
First, let $w$ be any alternative that is not included in the supernode dominated by $x$. 
In this case, there exist at least $\alpha$ disjoint paths from $x$ to $w$. 
This is because there exists a path from $x$ to $w$ containing at least three alternatives which uses only alternatives from supernodes of size $\alpha$ in its interior (in other words, all alternatives besides $x$ and $w$ belong to such supernodes) and does not use more than one alternative from the same supernode.
By construction, such a path gives rise to $\alpha$ edge-disjoint paths from $x$ to $w$.
Second, let $w$ be an alternative in the supernode dominated by $x$. 
Then, due to the structure of the supernode, there exist $(\alpha-1)/2$ disjoint two-step $x$-$w$-paths and one direct $x$-$w$-path, i.e., the edge $(x,w)$. Hence, the \mov{} of $x$ is at least $(\alpha+1)/2$. 
 \end{proof}
 As discussed earlier, this Claim concludes the proof of \Cref{thm:kkings-degree-consistency}. 
\end{proof}

\begin{corollary}
$\mov_{\cp}, \mov_{\uc}, \mov_{k\emph{-kings}}$ (for constant $k\ge 3$),  and $\mov_{\ba}$ do not fulfill strong degree-consistency.
\end{corollary}

\subsection{A Probabilistic Result}
\label{sec:probabilistic}

In this section, we establish a simple formula for the \mov{} of $\tc$ and $k$-kings for $k\ge 4$ that works ``with high probability'', i.e., the probability that the formula holds converges to $1$ as $n$ grows.
We assume that the tournament is generated using the \emph{uniform random model}, where each edge is oriented in either direction with equal probability independently of other edges; this model has been studied, among others, by \citet{Fey08} and \citet{ScottFe12}.

\begin{restatable}{theorem}{thmtchighprobability}\label{thm:tc-high-probability}
Let $S\in\{\tc, k\text{-kings}\}$, where $4\leq k\leq n-1$.
Assume that a tournament $T$ is generated according to the uniform random model.
Then, with high probability, the following holds for all $x\in V(T)$ simultaneously:
\[
\mov_S(x,T) = \min\left(\outdeg(x), \min_{y\in V(T): y\neq x}\indeg(y)\right).
\]
\end{restatable}

Theorem \ref{thm:tc-high-probability} suggests that when tournaments are generated according to the uniform random model, $\mov_{\tc}$ and $\mov_{k\text{-kings}}$ for $k\ge 4$ can likely be computed by a simple formula based on the degrees of the alternatives.
In particular, even though the problem is computationally hard for $\mov_{k\text{-kings}}$ for any constant $k\ge 4$ \citemovfull, there exists an efficient heuristic that correctly computes the \mov{} value in most cases.
In Appendix~\ref{app:tc-example}, we give an example showing that the heuristic is not always correct. 
More precisely, for any positive integer $\ell$, we construct a tournament such that $\{\mov_{\tc}(x,T) \mid x \in V(T)\}$ contains the values $1,2,\dots,\ell$ whereas the formula in \Cref{thm:tc-high-probability} predicts that all alternatives have the same (arbitrarily large) $\mov_{\tc}$ value.

At a high level, to prove this theorem, we first observe that by a result of \citet{Fey08}, it is likely that $S(T) = V(T)$, i.e., all alternatives are chosen by $S$.
In order to remove alternative $x$ from the winner set, one option is to make it a Condorcet loser---this requires $\outdeg(x)$ reversals---while another option is to make another alternative $y$ a Condorcet winner---this requires $\indeg(y)$ reversals.
Hence, the left-hand side is at most the right-hand side.
To establish that both sides are equal with high probability, we need to show that the aforementioned options are the best ones for making~$x$ a non-winner---by \Cref{lem:xcut-kkings-destr}, this requires making some~$y$ unreachable from $x$ in four steps.
The intuition behind this claim is that the tournament resulting from the uniform random model is highly connected, with many paths of length at most four from $x$ to $y$.
As a result, if we want to make $y$ unreachable from $x$, it is unlikely to be beneficial to destroy intermediate edges instead of edges adjacent to~$x$ or~$y$.

To prove the theorem, we first state the Chernoff bound, a standard tool for bounding the probability that the value of a random variable is far from its expectation.

\begin{lemma}[Chernoff bound] \label{lem:chernoff}
Let $X_1, \dots, X_k$ be independent random variables taking values in $[0, 1]$, and let $S := X_1 + \cdots + X_k$. Then, for any $\delta \geq 0$,
$$\Pr[S \geq (1 + \delta)\E[S]] \leq \exp\left(\frac{-\delta^2 \E[S]}{3}\right)$$
and
$$\Pr[S \leq (1 - \delta)\E[S]] \leq \exp\left(\frac{-\delta^2 \E[S]}{2}\right).$$
\end{lemma}

\begin{proof}[Proof of \Cref{thm:tc-high-probability}]
Let $r := n-1$, and consider the following three events:
\begin{enumerate}
\item $S(T)=V(T)$;
\item For every $x\in V(T)$, it holds that $ \outdeg(x),\indeg(x)\in[0.49r,0.51r]$;
\item For every pair of disjoint sets $A,B\subseteq V(T)$ such that $|A|,|B|\geq 0.1r$, the number of edges directed from an alternative in $A$ to an alternative in $B$ is at least $0.004r^2$.
\end{enumerate}

We claim that with high probability, all three events occur simultaneously. 
By union bound, it suffices to prove this claim for each event separately.
The claim for event~(i) follows from Theorem~1 of \citet{Fey08}, which shows that the Banks set includes all alternatives with high probability in a random tournament, along with the fact that in any tournament, the Banks set is contained in the uncovered set, which is in turn contained in our tournament solution $S$.

Fix $x\in V(T)$, and let $X_1,\dots,X_{n-1}$ be indicator random variables that indicate whether $x$ dominates each of the remaining $n-1$ alternatives or not; $X_i$ takes the value $1$ if so, and $0$ otherwise.
Let $X:=\sum_{i=1}^{n-1}X_i$.
We have $\E[X_i]=0.5$ for each $i$, and so $\E[X]=0.5r$.
By Lemma~\ref{lem:chernoff}, it follows that
$$
\Pr[X\geq 0.51r] \leq \exp\left(-\frac{0.02^2\cdot 0.5r}{3}\right)\leq \exp(10^{-5}r).
$$
Similarly, by applying the other inequality in Lemma~\ref{lem:chernoff}, we have $\Pr[X\leq 0.49r] \leq \exp(10^{-5}r)$.
Taking a union bound over these two events and over all $x\in V(T)$, the probability that $\outdeg(x)\not\in [0.49r,0.51r]$ for some $x$ is at most $2n\cdot \exp(10^{-5}r)\leq 4r\cdot \exp(10^{-5}r)$, which converges to $0$ as $r\rightarrow\infty$ (equivalently, as $n\rightarrow\infty$).
Since $\outdeg(x)+\indeg(x)=r$ for each $x$, having $\outdeg(x)\in[0.49r,0.51r]$ implies $\indeg(x)\in[0.49r,0.51r]$ as well.
This means that event (ii) occurs with high probability.

Next, fix a pair of disjoint sets $A,B\subseteq V(T)$ such that $|A|,|B|\geq 0.1r$.
Let $t$ be the number of edges between $a$ and $b$, and let $Y_1,\dots,Y_t$ be indicator random variables that indicate whether each edge is oriented from $A$ to $B$; $Y_i$ takes the value $1$ if so, and $0$ otherwise.
Let $Y := \sum_{i=1}^t Y_i$.
We have $\E[Y_i]=0.5$ for each $i$, and so $\E[Y]=0.5t$.
Writing $t=cr^2$ for some $c\geq 0.01$, it follows by Lemma~\ref{lem:chernoff} that
\begin{align*}
\Pr[Y\leq 0.004r^2]
&= \Pr\left[Y\leq\frac{0.004}{0.5c}\cdot\E[Y]\right] \\
&\leq \Pr[Y\leq 0.8\cdot \E[Y]] \\
&\leq \exp(-0.02\cdot\E[Y])
\leq \exp(-10^{-4}r^2).
\end{align*}
Since there are no more than $2^n$ choices for each of $A$ and $B$, by union bound, the probability that event (iii) fails for some pair $A,B$ is at most $2^{2n}\cdot \exp(-10^{-4}r^2)\leq \exp(4r-10^{-4}r^2)$, which again vanishes for large $r$.
We have therefore established that events (i), (ii), and (iii) occur simultaneously with high probability.

Assume from now on that all three events occur, and let $r\geq 130$.
We will show that under these conditions, it always holds that $$\mov_S(x,T) = \min\left(\outdeg(x), \min_{y\in V(T): y\neq x}\indeg(y)\right).$$
This suffices to finish the proof of the theorem.

First, since event (i) occurs, $\mov_S(x,T)$ is positive for every $x\in V(T)$.
We claim that for any distinct $x,y\in V(T)$, it holds that \begin{equation}
\label{eq:mincut}
\mincut_k(x,y) = \min\left(\outdeg(x),\indeg(y)\right).  
\end{equation}
If (\ref{eq:mincut}) holds, we would have that the size of a minimum $k$-length bounded $x$-cut is
\begin{align*}
\min_{y\neq x} (\mincut_k(x,y)) 
&= 
\min_{y\neq x} (\min\left(\outdeg(x),\indeg(y)\right)) \\
&= \min\left(\outdeg(x), \min_{y\neq x}\indeg(y)\right).
\end{align*}
By Lemma~\ref{lem:xcut-kkings-destr}, this size is equal to the size of a minimum DRS for $x$ with respect to $S$, i.e., $\mov_S(x,T)$.
To finish the proof, it therefore remains to establish (\ref{eq:mincut}).

Fix a pair $x,y\in V$.
Observe that the following two sets are $k$-length bounded $x$-$y$-cuts:
\begin{itemize}
\item The set of all outgoing edges from $x$ (since $x$ cannot reach any other alternative upon the removal of these edges);
\item The set of all incoming edges into $y$ (since $y$ cannot be reached by any other alternative upon the removal of these edges).
\end{itemize}
The former set has size $\outdeg(x)$ and the latter set has size $\indeg(y)$, implying that $\mincut_k(x,y) \leq \min\left(\outdeg(x),\indeg(y)\right)$. 

Assume now for the sake of contradiction that this inequality is strict, i.e., there exists a $k$-length bounded $x$-$y$-cut $R$ of size less than $\min\left(\outdeg(x),\indeg(y)\right)$.
Let $E(x,D(x))$ denote the set of edges between $x$ and its dominion $D(x)$, and let $E(y,\overline{D}(y))$ denote the set of edges between $y$ and its set of dominators $\overline{D}(y)$.
Since event (ii) occurs, we have $\left|E(x,D(x))\right|,\left|E(y,\overline{D}(y))\right|\in [0.49r,0.51r]$. 
Moreover, $|R|\leq 0.51r$.
Let $V_x\subseteq D(x)$ be the set of alternatives that $x$ can still directly reach after the removal of $R$.
Similarly, let $V_y\subseteq \overline{D}(y)$ be the set of alternatives that can directly reach $y$ after the removal of $R$.
We consider two cases:

\emph{Case 1:} $R$ contains at most $0.39r$ edges in each of $E(x,D(x))$ and $E(y,\overline{D}(y))$.
This means that $|V_x|,|V_y|\geq 0.1r$.
If $V_x\cap V_y\neq\emptyset$, then $x$ can reach $y$ via a path of length two even after the removal of $R$, a contradiction.
So $V_x$ and $V_y$ must be disjoint.
Since event (iii) occurs, there are at least $0.004r^2$ edges directed from an alternative in $V_x$ to an alternative in $V_y$.
In addition, from $r\geq 130$ we have $0.004r^2 > 0.51r$, so at least one of these edges is not included in $R$.
It follows that after $R$ is removed, there still exists a path of length three from $x$ to $y$, a contradiction.

\emph{Case 2:} $R$ contains at least $0.39r$ edges in either $E(x,D(x))$ or $E(y,\overline{D}(y))$.
Assume without loss of generality that it contains at least $0.39r$ edges in $E(x,D(x))$; the other case can be handled analogously.
Since $|R|\leq 0.51r$, $R$ contains at most $0.12r$ edges in $E(y,\overline{D}(y))$.
So $|V_y|\geq 0.49r-0.12r = 0.37r$.
Now, since $|R| < \min(\outdeg(x),\indeg(y))\leq \outdeg(x)$, we have $|V_x|\geq 1$.
Let $z$ be an arbitrary alternative in $V_x$, and let $E(z,D(z))$ denote the set of edges between $z$ and its dominion $D(z)$.
Let $V_z\subseteq D(z)$ be the set of alternatives that $z$ can reach directly after $R$ is removed.
Repeating our argument for $V_y$, we get $|V_z|\geq 0.37r$.

The rest of the argument in Case~2 mirrors that of Case~1, with $V_z$ taking the role of $V_x$.
If $V_z\cap V_y\neq\emptyset$, then $x$ can reach $y$ via a path of length three even after the removal of $R$, a contradiction.
Else, $V_z$ and $V_y$ are disjoint.
Since event (iii) occurs, there are at least $0.004r^2$ edges directed from an alternative in $V_z$ to an alternative in $V_y$, so at least one of these edges is not included in $R$.
It follows that after $R$ is removed, there still exists a path of length four from $x$ to $y$, a contradiction.

It follows that we reach a contradiction in both cases, and the proof is complete.
\end{proof}

\section{Experiments}
\label{sec:experiments}

In order to better understand how \mov{} values of tournament solutions behave in practice, we conducted computational experiments using randomly generated tournaments. For the sake of diversity of the generated instances, we implemented six different stochastic models to generate tournaments. 
To make our study comparable to the experiments presented by \citet{BrandtSe16}, we selected a similar set of stochastic models and parameterizations. 

Given a tournament solution $S$ and a tournament $T$, we are interested in 
\begin{itemize}
\item the number $|\arg \max_{x \in V(T)} \mov_S(x,T)|$ of alternatives with maximum $\mov_S{}$ value, and
\item the number $|\{\mov_S(x,T): x \in V(T)\}|$ of different \mov{} values taken by all alternatives in the tournament.
\end{itemize}
The first value directly measures the discriminative power of the refinement of $S$ that only selects alternatives with a maximal $\mov_S$ value, whereas the second value measures more generally the ability of the \mov{} notion to distinguish between the alternatives in a tournament. 

\paragraph{Set-up}
We used six stochastic models to generate preferences: 
the uniform random model (which was used in \Cref{sec:probabilistic}),
two variants of the \emph{Condorcet noise model} (with and without voters), 
the \emph{impartial culture} model,
the P\'olya-Eggenberger \textit{urn model}, and the
\textit{Mallows} model. 

We first describe two models that directly create tournaments without creating a preference profile of a set of voters beforehand. 
The simplest way to create a tournament is to start with a complete undirected graph and decide the direction of each edge independently by flipping a fair coin---we call this the \emph{uniform random model}. 
The \emph{Condorcet noise model} is similar but slightly more biased: Here, we start with an initial order $\succ$ on the alternatives and some fixed parameter $1/2\leq p \leq 1$. Then, for two different alternatives $a$ and $b$ where $a \succ b$, the edge $(a,b)$ is included in the tournament with probability $p$; otherwise, the edge $(b,a)$ is included. 

For the remaining four stochastic models, we first create a preference profile of a set of voters, i.e., each voter has a complete and antisymmetric (but not necessarily transitive) preference relation over the set of alternatives. 
Like \citet{BrandtSe16}, we set the number of voters to $51$. 
Then, we consider the majority relation, which induces a tournament when there are an odd number of voters. 
One way to generate a preference profile is similar to the previously discussed Condorcet noise model, i.e., the \emph{Condorcet noise model with voters}. 
Again, we start with a random order $\succ$ on the alternatives and some fixed parameter $1/2\leq p \leq 1$. Now, the preference relation for each voter is created just as we created the tournament in the Condorcet noise model, before we take the majority among these preferences.

The other three stochastic models all assign a ranking to each voter, i.e., the individual preference relations are now required to be transitive. 
In the \emph{impartial culture} model, for each voter, the probability of obtaining a ranking is uniformly distributed over all possible rankings and thereby independent of the selection for other voters. 
A similar but more correlated way to select rankings is the P\'olya-Eggenberger \emph{urn model}, suggested by \citet{Berg85a}. 
For this model, imagine an urn which initially contains each possible ranking exactly once. 
Then, after each voter has drawn a ranking from the urn, the ranking is placed back together with $\alpha$ copies of it. Naturally, the parameter $\alpha$ controls the degree of similarity among the voters. 
Lastly, we also applied the \emph{Mallows} model \citep{Mall57a}. 
Assuming a ground truth ranking, the probability that a voter is assigned a particular ranking in this model grows when the Kendall tau distance to the ground truth ranking becomes smaller. 
The dispersion parameter $\phi \in (0,1]$ controls the concentration of the probability mass on rankings that are close to the ground truth ranking. 
More precisely, $\phi=1$ corresponds to the uniform distribution over all possible rankings, while $\phi \rightarrow 0$ concentrates more probability on the ground truth ranking and rankings close to it.

For each stochastic model and each number of alternatives $n \in \{5,10,15,20,25,30\}$, we sampled $100$ tournaments. 
Using the methods described by \citemov, we implemented algorithms to calculate the \mov{} values for \cp, \uc, $3$-kings, and \tc.
Due to their computational intractability, we did not implement procedures to calculate the \mov{} values for \ba and $k$-kings for $k \geq 4$. 

The experiments were carried out on a system with 1.4 GHz Quad-Core Intel Core i5 CPU, 8GB RAM,
and macOS 10.15.2 operating system. 
The software was implemented in Python 3.7.7 and the libraries networkx 2.4, matplotlib 3.2.1, numpy 1.18.2, and pandas 1.0.3 were used.
For implementing the Mallows and urn models, we utilized implementations contributed by \citet{MaWa13a}.
The code for our implementation can be found at \url{http://github.com/uschmidtk/MoV}.

\begin{figure*}[!ht]
\begin{minipage}{\textwidth}
\center{\textbf{Average Size of Maximum Equivalence Class}}
\end{minipage}
\vspace{.5cm}
\begin{minipage}{0.5\textwidth}
\includegraphics[width=\textwidth]{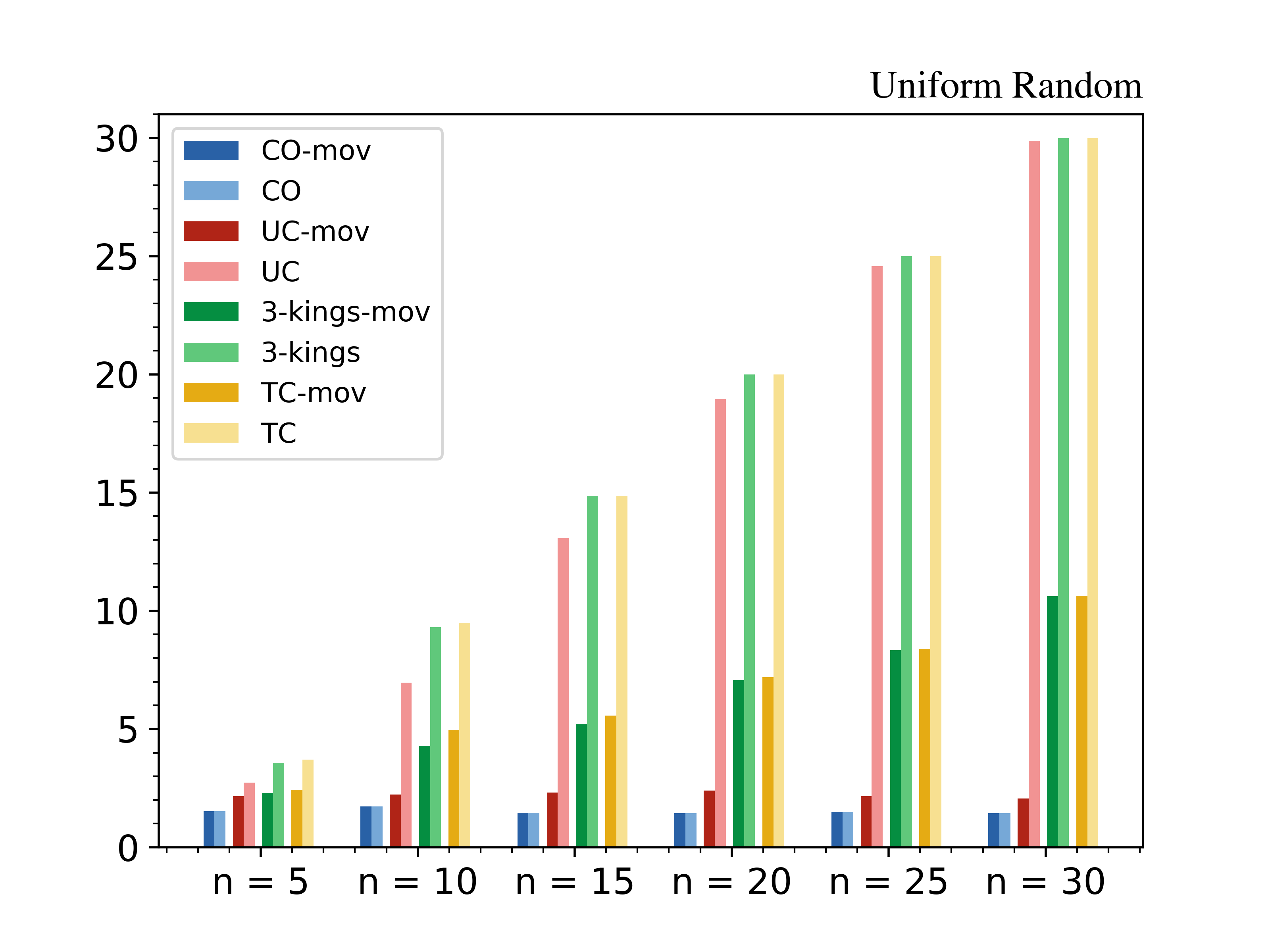}
\end{minipage}\begin{minipage}{0.5\textwidth}
\includegraphics[width=\textwidth]{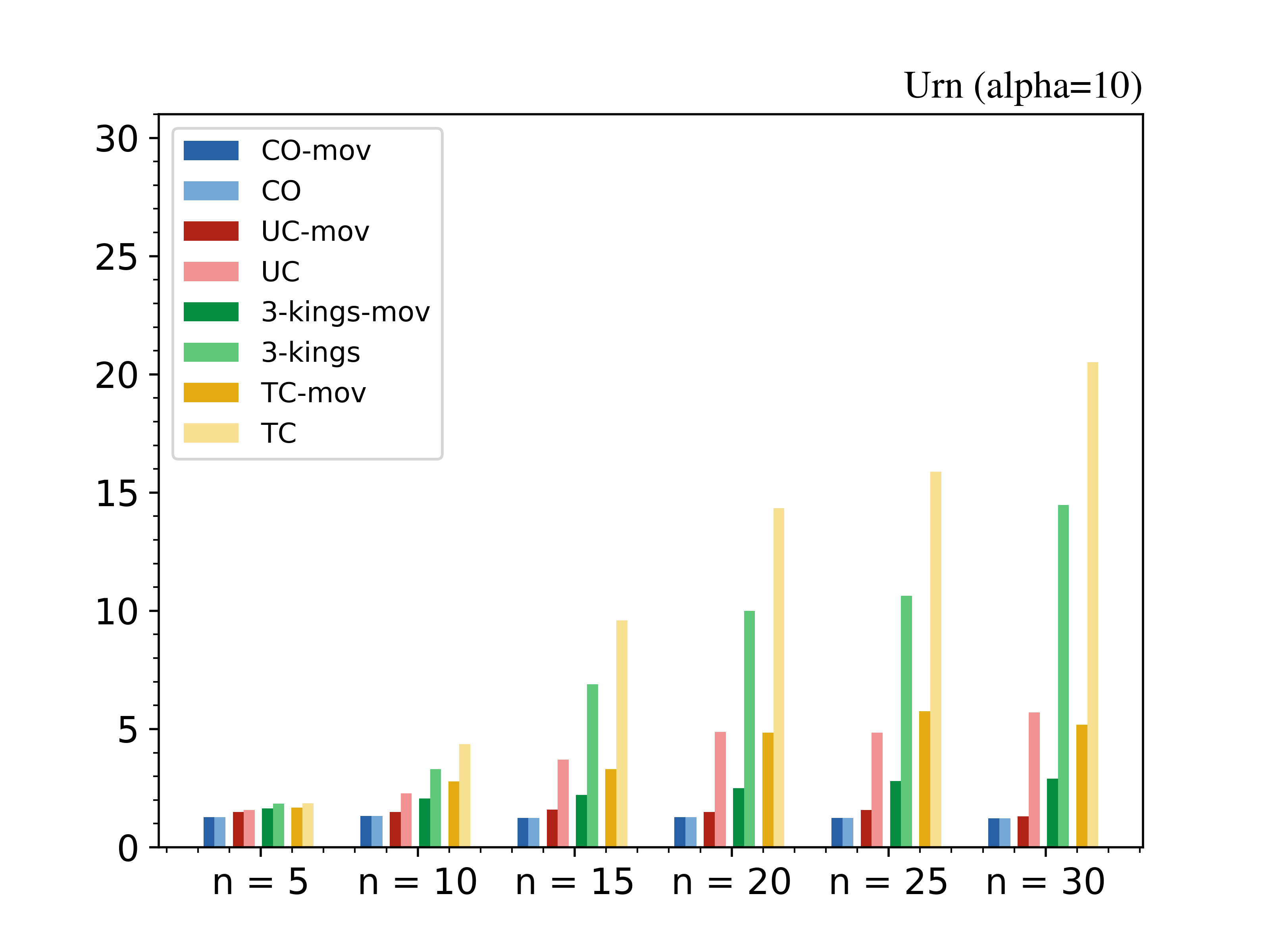}
\end{minipage}

\begin{minipage}{0.5\textwidth}
\includegraphics[width=\textwidth]{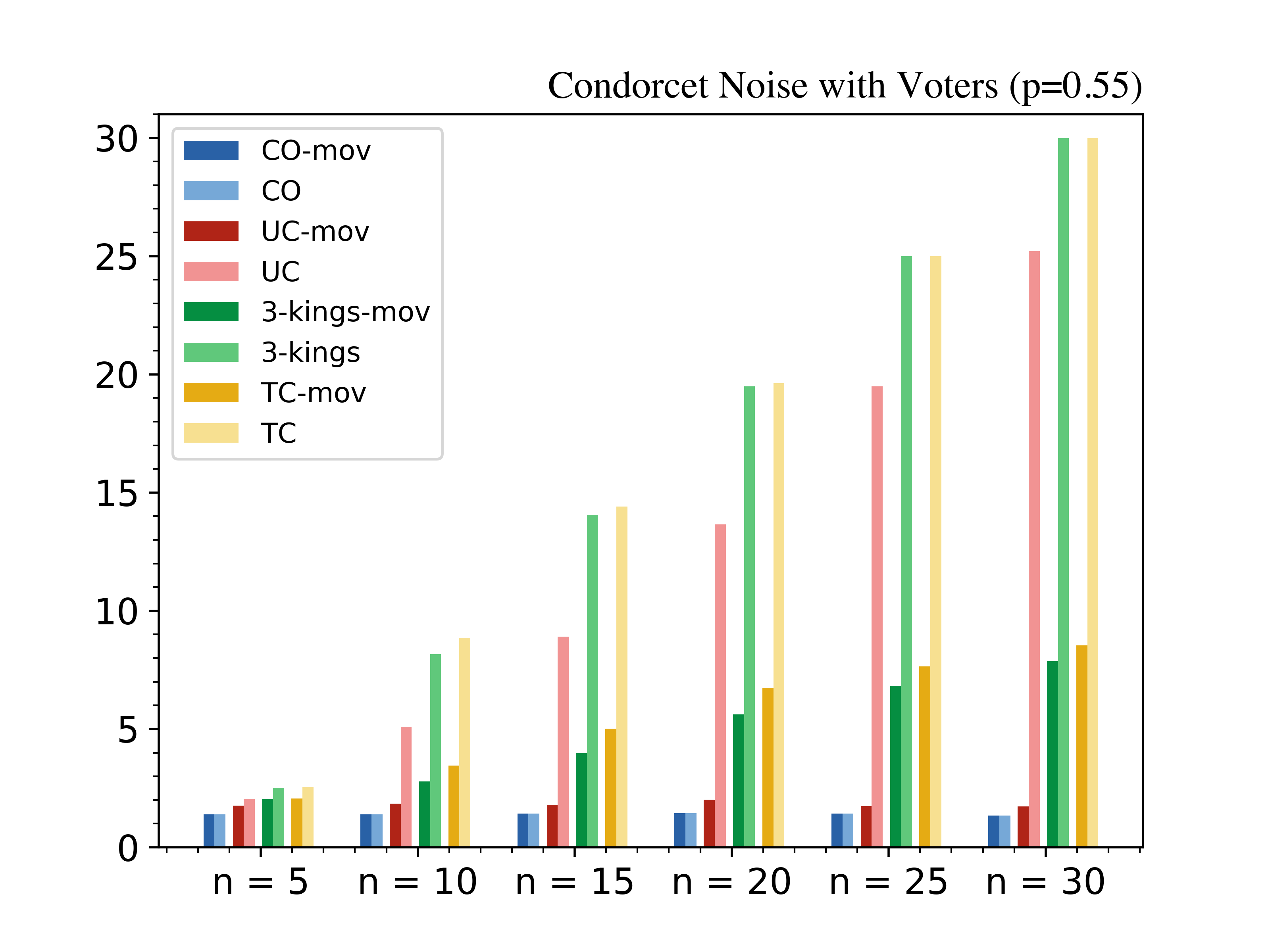}
\end{minipage}
\begin{minipage}{0.5\textwidth}
\includegraphics[width=\textwidth]{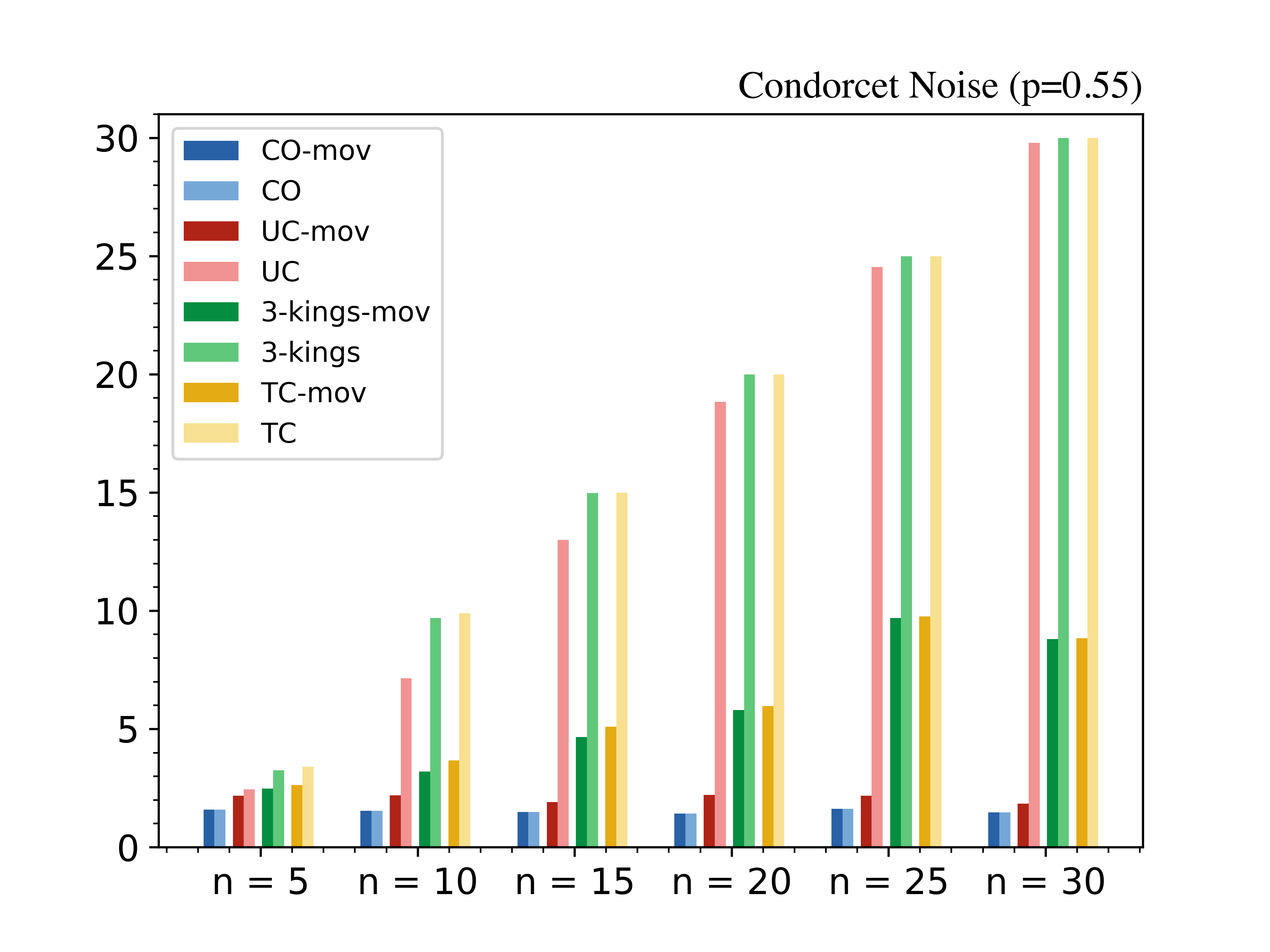}
\end{minipage}

\begin{minipage}{0.5\textwidth}
\includegraphics[width=\textwidth]{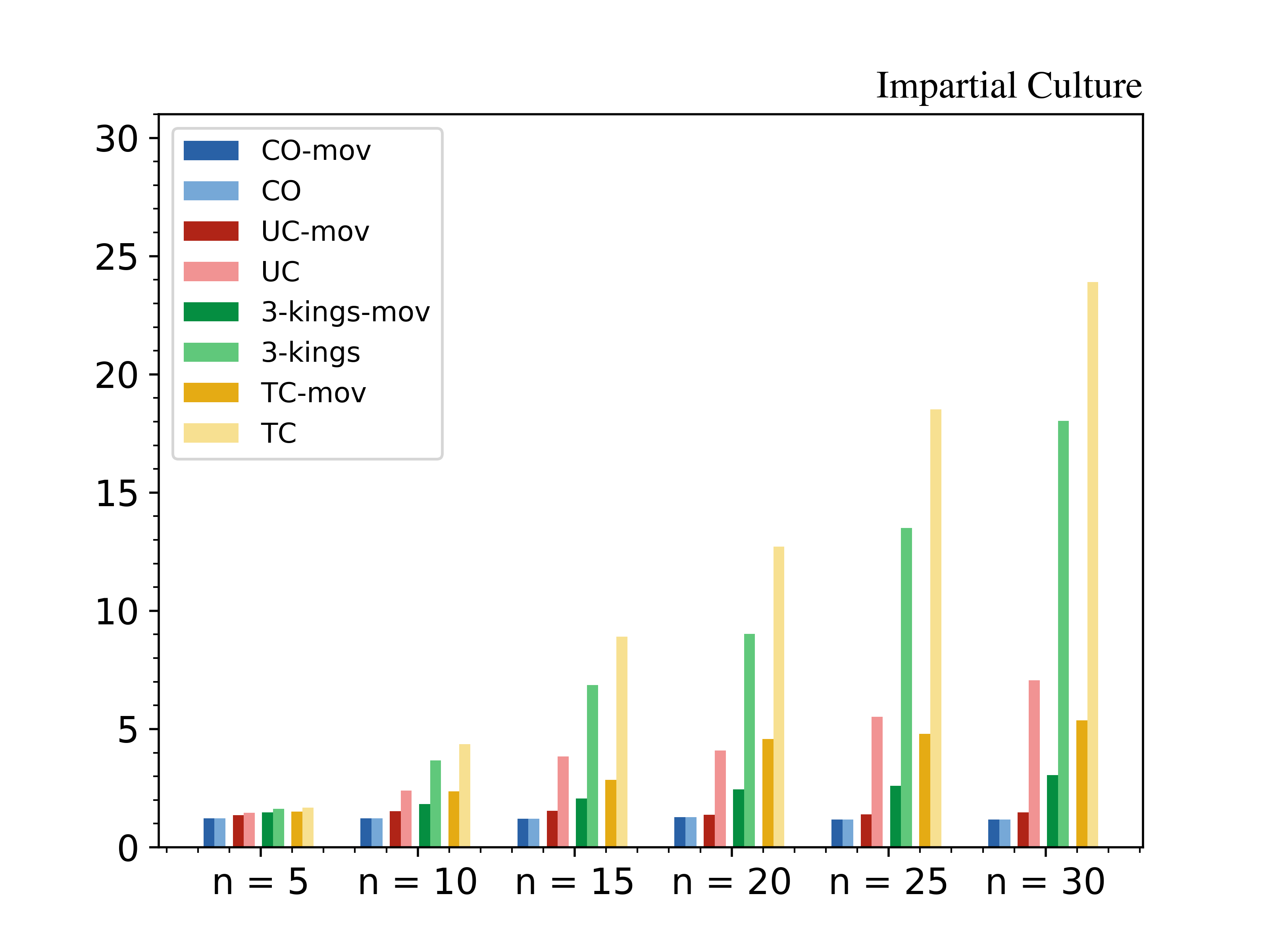}
\end{minipage}
\begin{minipage}{0.5\textwidth}
\includegraphics[width=\textwidth]{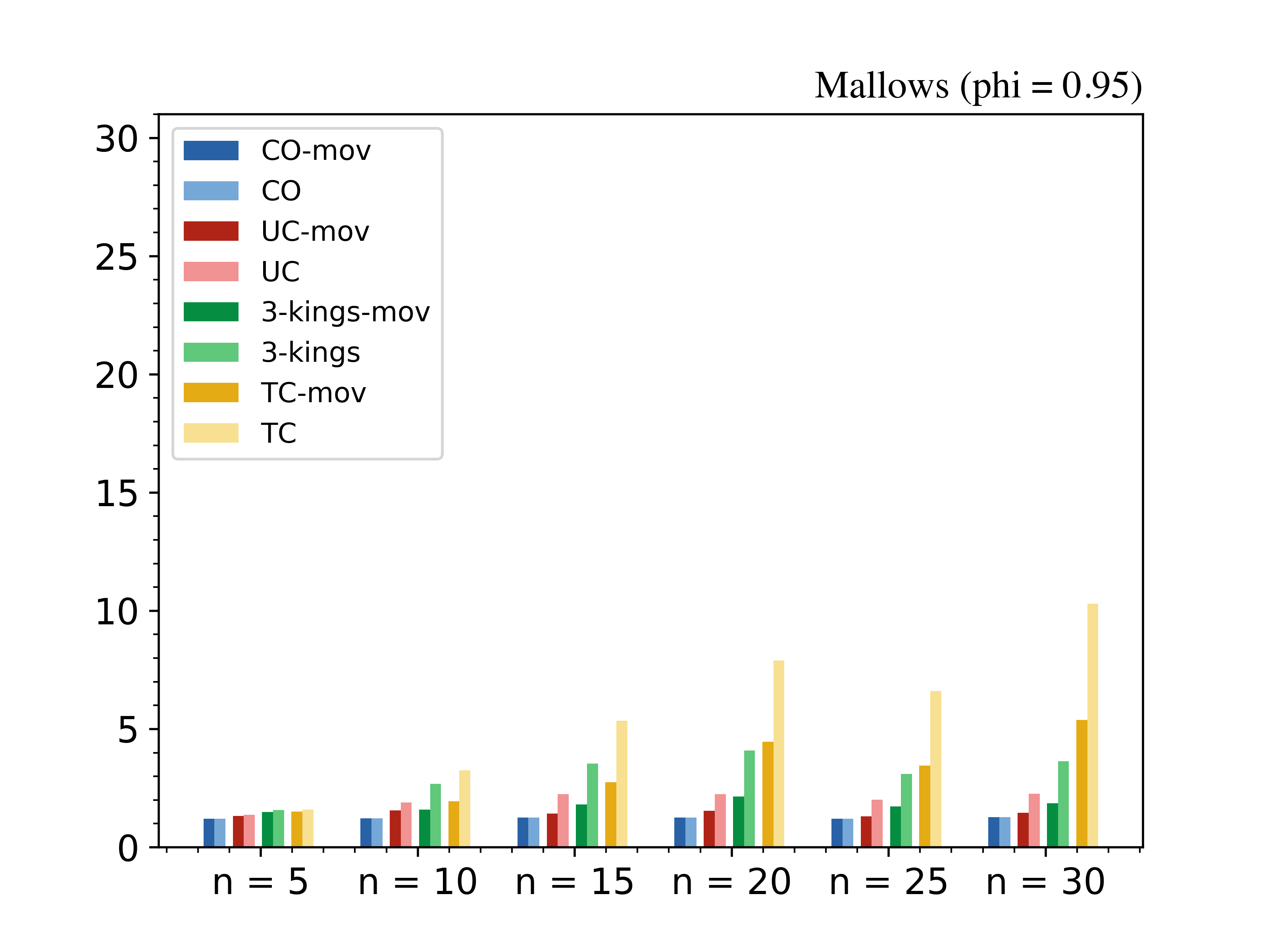}
\end{minipage}
\caption{The illustrations show the average number of alternatives with maximum \mov{} value for different stochastic models, tournament solutions and sizes. 
For comparison, the average size of the entire winning set of the corresponding original tournament solution is depicted by a lighter shade.}\label{fig:experiments-max}
\end{figure*}

\begin{figure*}[!ht]
\begin{minipage}{\textwidth}
\center{\textbf{Average Number of Unique Values}}
\end{minipage}
\vspace{.5cm}
\begin{minipage}{0.5\textwidth}
\includegraphics[width=\textwidth]{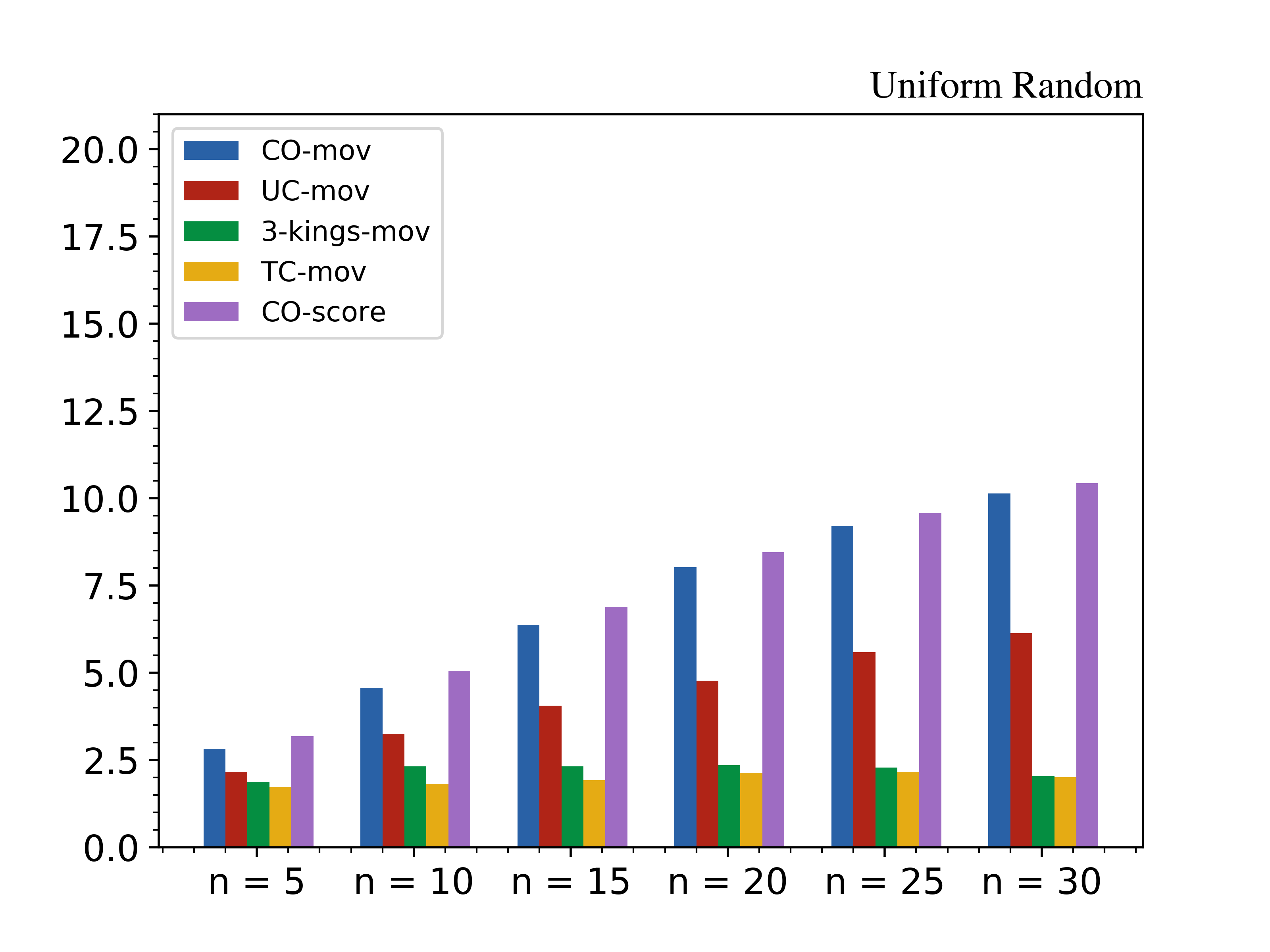}
\end{minipage}\begin{minipage}{0.5\textwidth}
\includegraphics[width=\textwidth]{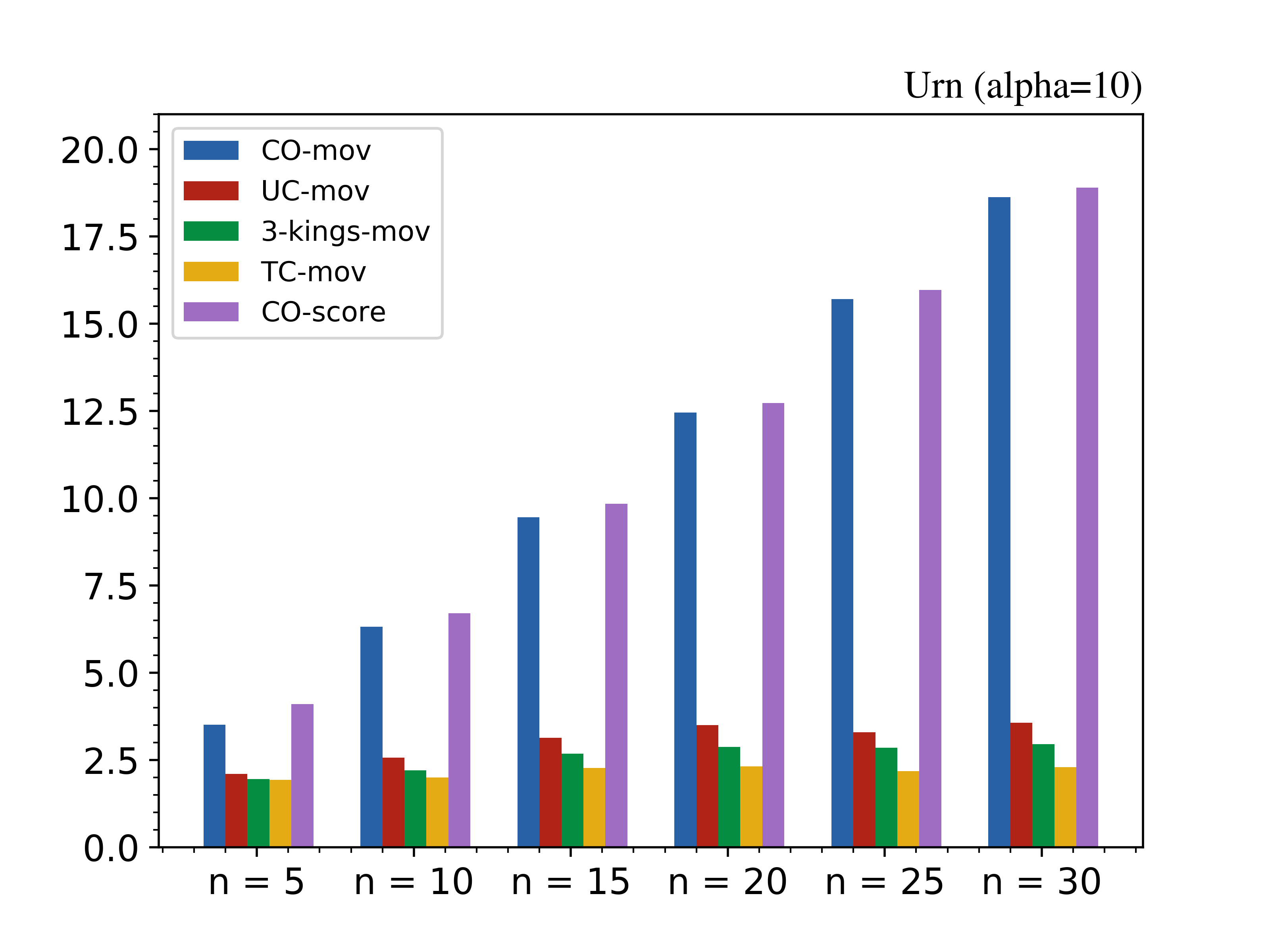}
\end{minipage}

\begin{minipage}{0.5\textwidth}
\includegraphics[width=\textwidth]{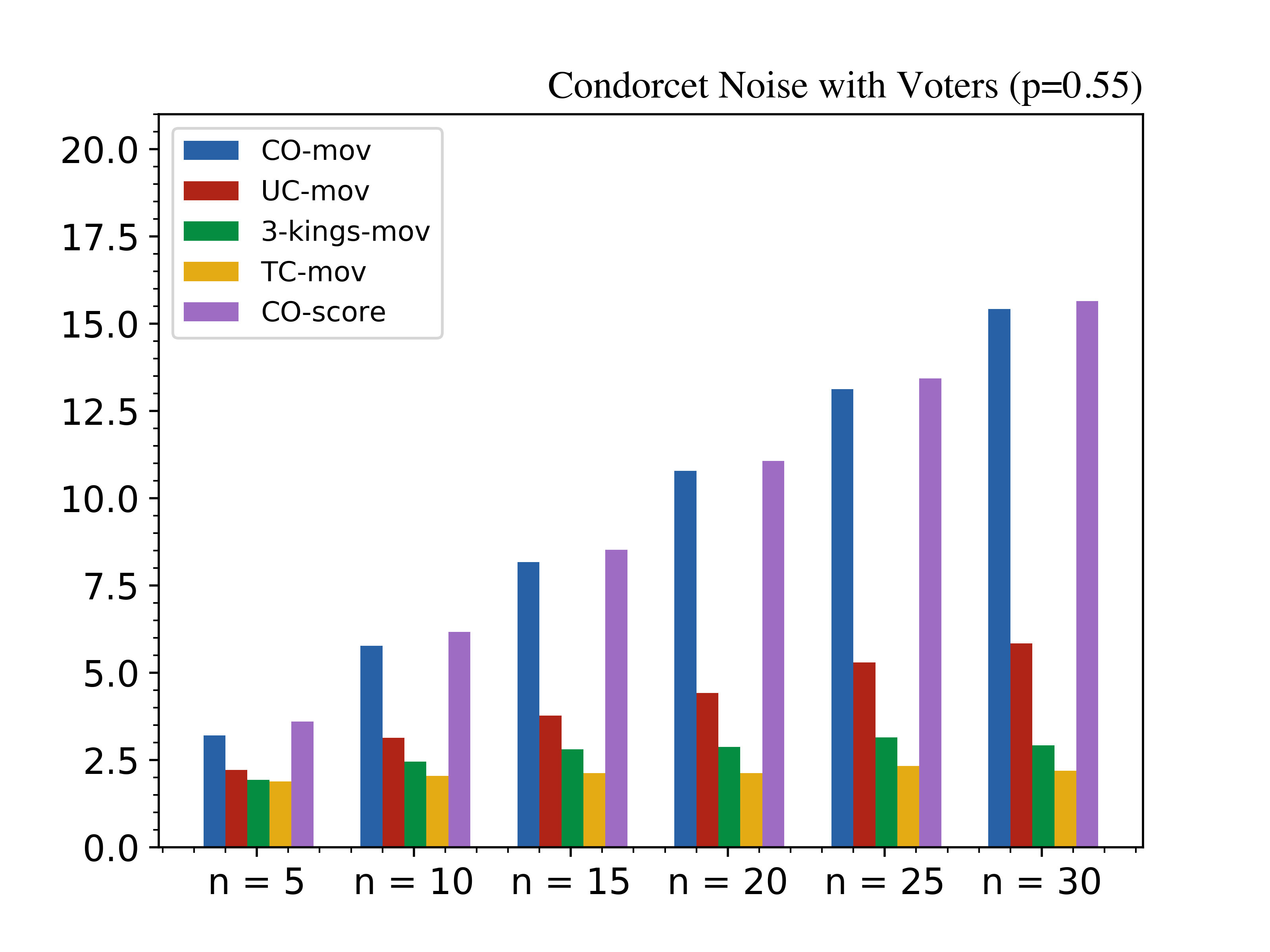}
\end{minipage}
\begin{minipage}{0.5\textwidth}
\includegraphics[width=\textwidth]{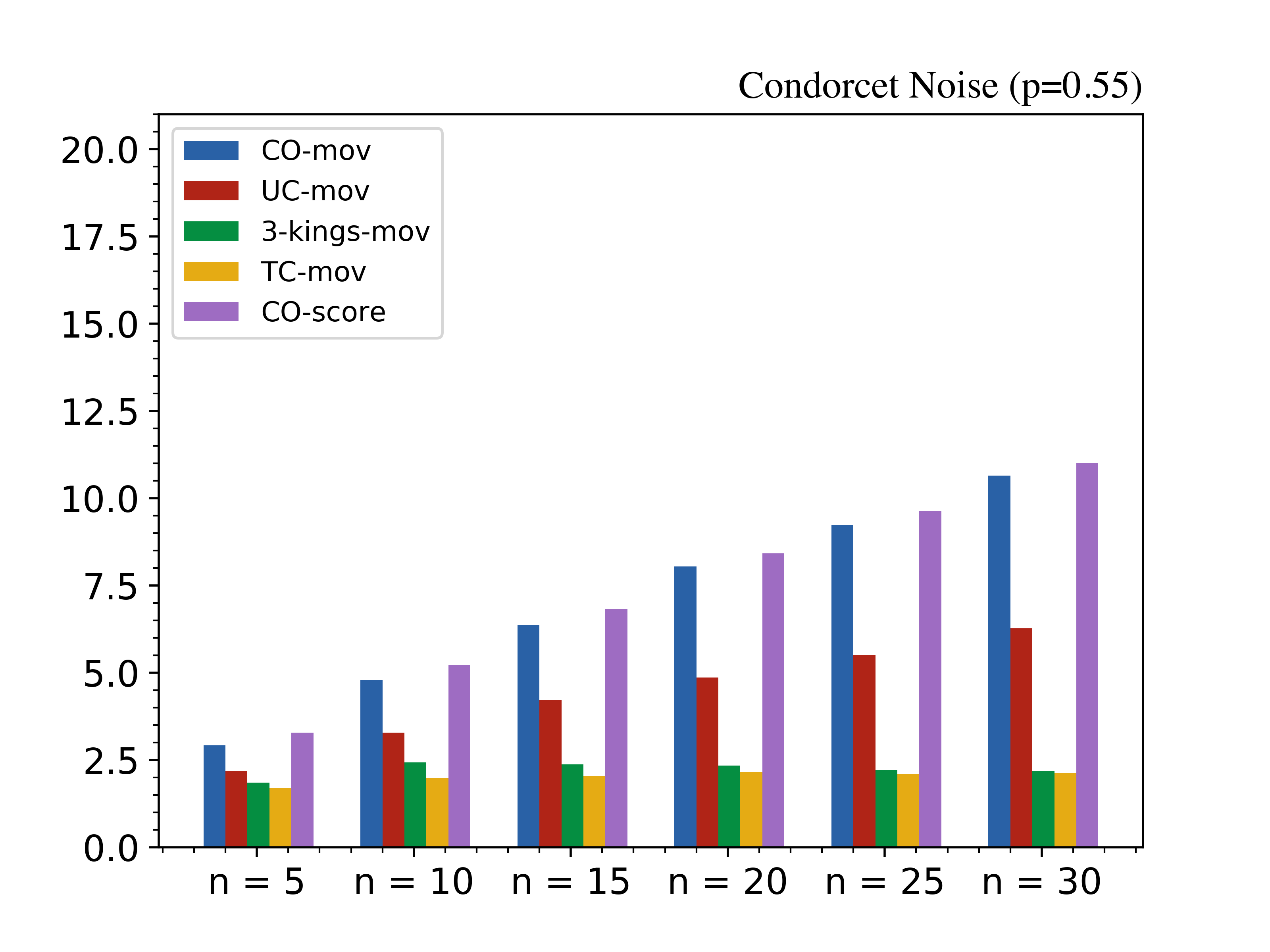}
\end{minipage}

\begin{minipage}{0.5\textwidth}
\includegraphics[width=\textwidth]{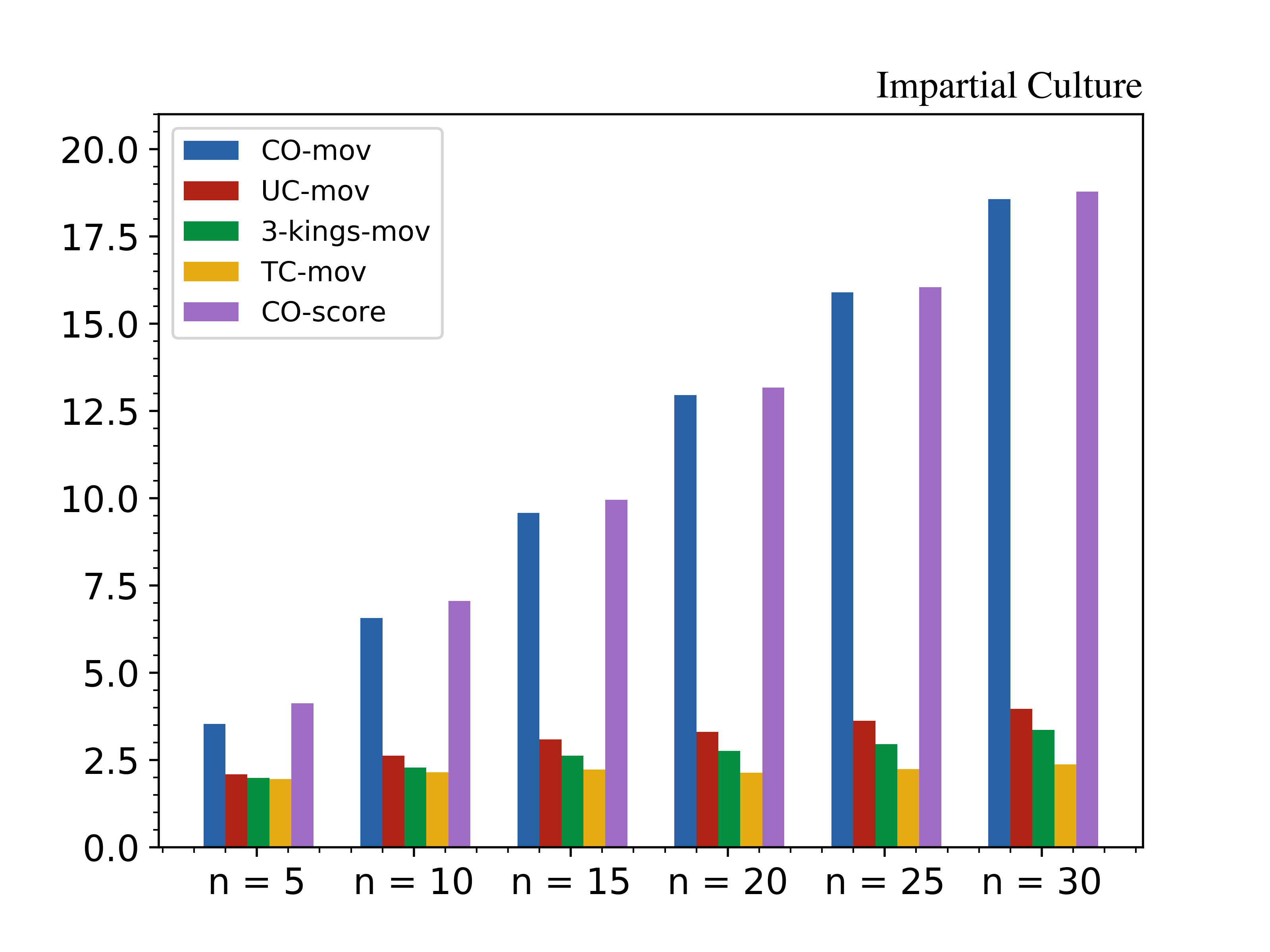}
\end{minipage}
\begin{minipage}{0.5\textwidth}
\includegraphics[width=\textwidth]{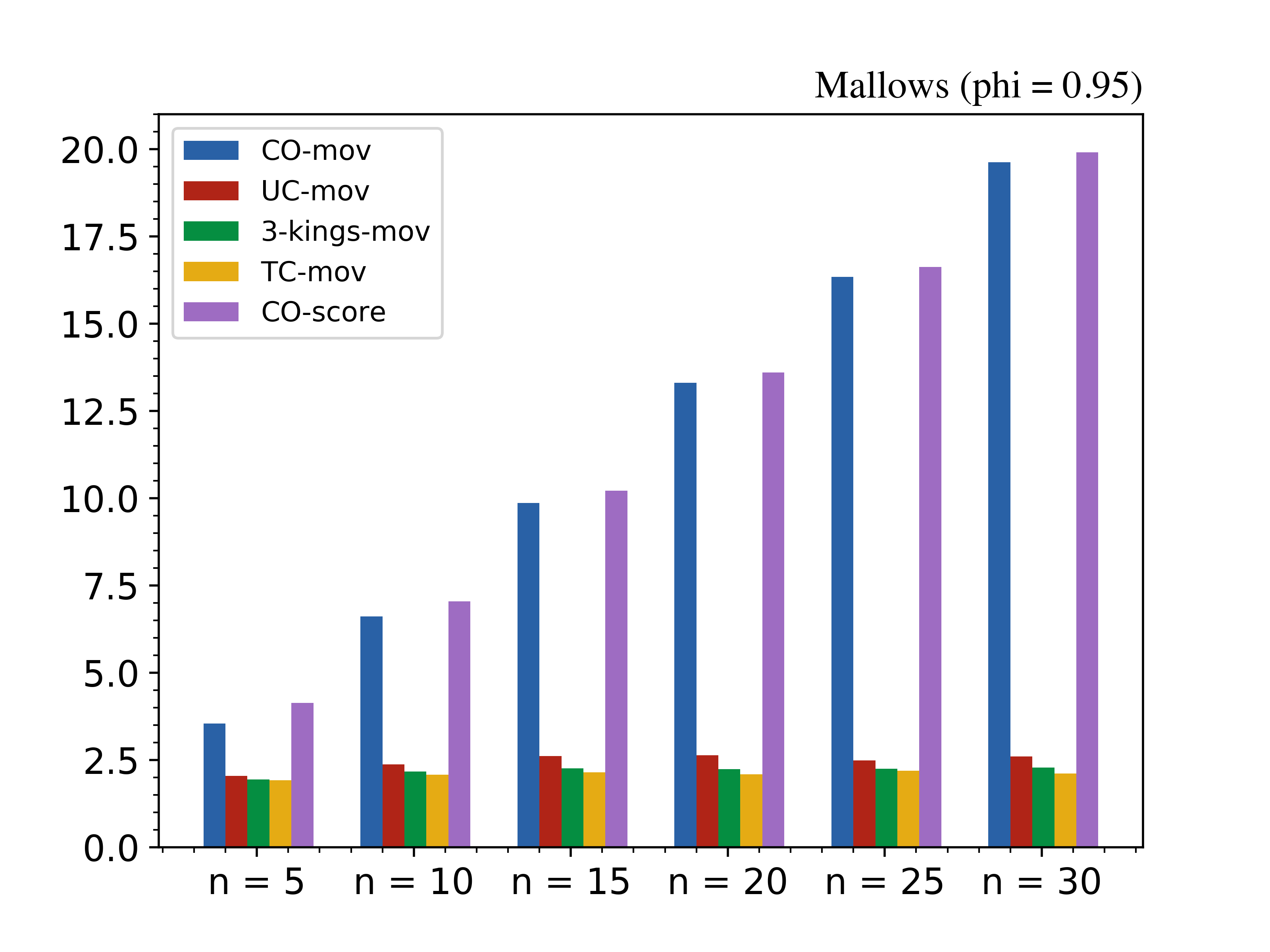}
\end{minipage}
\caption{The illustrations show the average number of unique \mov{} values for different stochastic models, tournament solutions and sizes.  
For comparison, the average number of unique Copeland scores is shown in violet. 
} \label{fig:experiments-unique}
\end{figure*}

\paragraph{Results}
Figure \ref{fig:experiments-max} depicts the average size of the set of alternatives with maximum \mov{} value, and Figure \ref{fig:experiments-unique} shows the average number of unique \mov{} values. 

\paragraph{Observations}

The first observation we make is that $\mov_{3\text{-kings}}$ behaves rather similarly to $\mov_{\tc}$: the average number of alternatives with maximum \mov{} grows with increasing $n$, and this number is on average slightly less than half of the number of $3$-kings and \tc winners, respectively. 
However, this ratio becomes smaller for tournaments where the number of $3$-kings or \tc winners is already large. 
For example, when we only consider tournaments where the number of \tc{} winners is greater than $10$, only one-third of the \tc winners have a maximum $\mov_{\tc}$ value on average; the same holds for $3$-kings. 
However, a more detailed look at the experimental results show that for both $3$-kings and \tc, the set of alternatives with maximum $\mov{}$ consists of only one alternative in around 73\%  of all instances, while in the remaining instances this set is typically large. 
This particular behavior for \tc and the uniform random model can be explained by \Cref{thm:tc-high-probability}: With high probability, the \mov{} values for \tc winners follow a specific formula based on the degrees, which leads to the set of alternatives with maximum \mov{} containing either a single alternative or a large number of alternatives in most cases.\footnote{Indeed, if there is a unique Copeland winner, that winner will be the unique alternative with the largest \mov{} according to the formula. Otherwise, for several alternatives (including the Copeland winners), it can be the case that their \mov{} is equal to $\indeg(y)$ for a Copeland winner $y$.} 
Our experiments show that this behavior is also present in tournaments generated by other stochastic models as well as for $3$-kings; formalizing the behavior theoretically is an interesting future direction.

Our second main observation is that $\mov_{\uc}$ behaves quite differently from $\mov_{3\text{-kings}}$ and $\mov_\tc$. Most importantly, the number of \uc winners with maximum $\mov_{\uc}$ does \textit{not} increase with a growing number of alternatives, but remains more or less constant for each stochastic model. 
For the uniform random model and the Condorcet noise models, this value is around $2$, while it is roughly $1.4$ for Mallows, the urn model, and the impartial culture model. 
As can be seen in \Cref{fig:experiments-max}, the set of alternatives maximizing $\mov_{\uc}$ is almost as discriminative as the Copeland set (all of whose alternatives maximize $\mov_{\cp}$).
However, we observe in \Cref{fig:experiments-unique} that the number of unique values of the Copeland score is notably higher than that of $\mov_{\uc}$. 
The latter is particularly low for models which tend to create tournaments with small \uc, including Mallows, impartial culture and the urn model. 
Both of these effects can be explained by the observation that $\mov_{\uc}$ is significantly better at distinguishing between \uc winners than it is at distinguishing between \uc non-winners.\footnote{\citemov{} showed that the smallest $\mov_{\uc}$ value in a tournament is bounded below by $-\lceil\log_2(n)\rceil$, and that this bound is asymptotically tight.
In our experiments, we observed that in most generated tournaments, the smallest $\mov_{\uc}$ value is much higher than this lower bound, namely either $-1$ or $-2$.} 
As a consequence, tournaments with a small uncovered set generally give rise to a small number of unique $\mov_{\uc}$ values.

\section{Discussion}

The recently introduced notion of margin of victory (\mov) provides a generic framework for refining any tournament solution. 
In this paper, we have contributed to the understanding of the \mov{} by providing not only structural insights but also experimental evidence regarding the extent to which it refines winner sets in stochastically generated tournaments. 
We established that the $\mov$ is consistent with the covering relation for all considered tournament solutions.
Moreover, we have identified a number of tournament solutions, including the uncovered set and the Banks set, for which the corresponding \mov{} values give insights into the structure of the tournament that go beyond simply comparing the outdegrees of alternatives, as witnessed by the fact that these \mov{} functions do not satisfy degree-consistency. 

\newcommand{\maxmovuc}{max-$\mov_{\uc}$\xspace}

In our experiments, the \mov{} function corresponding to the uncovered set ($\uc$) stands out for its discriminative power: not only is the set \maxmovuc (containing all alternatives with maximal $\mov_{\uc}$ score) consistently small, but the number of distinct $\mov_{\uc}$ scores is also relatively high in general. 
It is consequently tempting to suggest \maxmovuc as a new tournament solution. Besides its discriminative power and structural appeal, it can be computed efficiently (Brill et al. 2020) and inherits Pareto optimality from the uncovered set, which it refines \cite{BGH14a}. However, a thorough axiomatic analysis of \maxmovuc, as well as max-$\mov_{S}$\xspace for other tournament solutions $S$, is still outstanding. 

For tournaments with several highest-scoring alternatives (\ie several alternatives whose minimal destructive reversal sets are of the same maximal size), the \textit{number} of distinct destructive reversal sets may serve as a further criterion for distinguishing between winners. 
It would therefore be interesting to determine the complexity of computing such numbers, and also to study the size of the resulting refined winner set experimentally in future work.

\section*{Acknowledgments}

This work was partially supported by the Deutsche
Forschungsgemeinschaft under grant BR 4744/2-1, by
the European Research Council (ERC) under grant number
639945 (ACCORD), and by an NUS Start-up Grant. We would like to thank the anonymous reviewers for their valuable feedback.


\bibliography{abb,main}


\appendix

\section{Additional Results}

\subsection{Cover-Consistency}
\label{app:cover-consistency}

In the following two propositions, we show that neither monotonicity nor transfer-monotonicity can be dropped from the condition of Lemma~\ref{lem:cover-consistency}.
This also means that neither of the two properties implies the other.

\begin{proposition}
\label{prop:monotonic-cover-consistency}
There exists a monotonic tournament solution $S$ such that $\mov_S$ does not satisfy cover-consistency.
\end{proposition}

\begin{proof}
Let $S$ be a tournament solution such that an alternative $x$ is excluded if and only if it is dominated by an alternative of outdegree $1$ and the tournament has size at least four.\footnote{The latter condition is needed to ensure that the choice set is always nonempty.}
Suppose that an excluded alternative $x$ is dominated by an alternative $y$ of outdegree $1$. 
If $x$ becomes dominated by an additional alternative, $x$ remains dominated by $y$, whose outdegree is still $1$, so it remains excluded. Hence $S$ is monotonic.

To see that $\mov_S$ does not satisfy cover-consistency, consider a tournament $T$ composed of a regular tournament $T'$ of size $2k+1\geq 7$, along with an additional alternative $x$ which is dominated by all alternatives in $T'$.
Note that $S(T) = V(T)$.
For any $y\in T'$, the edge $(y,x)$ alone constitutes a DRS for $y$, so $\mov_S(y,T) = 1$.
On the other hand, since every alternative in $T'$ has outdegree at least $3$, in order for $x$ to be dominated by an alternative of outdegree $1$, at least two edges need to be reversed.
This means that $\mov_S(x,T) > \mov_S(y,T)$ even though $y$ covers $x$.
\end{proof}

\begin{proposition}
\label{prop:transfer-cover-consistency}
There exists a transfer-monotonic tournament solution $S$ such that $\mov_S$ does not satisfy cover-consistency.
\end{proposition}

\begin{proof}
Let $S$ be a tournament solution such that an alternative $x$ is excluded if and only if it has outdegree $1$ and there is another alternative of outdegree $0$ (so, in particular, $D(x)$ consists only of that alternative).
Assume that a tournament $T'$ is obtained by reversing edges $(y,z)$ and $(z,x)$ in a tournament $T$.
Note that in $T'$, if $x$ has outdegree $1$, then it dominates only $z$ which does not have outdegree $0$, so $x\in S(T')$ regardless of whether $x\in S(T)$.
Hence $S$ is transfer-monotonic.

To see that $\mov_S$ does not satisfy cover-consistency, consider any transitive tournament $T$.
Let $x$ and $y$ be the alternative of outdegree $1$ and $0$, respectively.
Then $\mov_S(x,T) < 0 < \mov_S(y,T)$ even though $x$ covers $y$.
\end{proof}

\subsection{Example showing that $\mov_{\tc}$ does not always follow the formula in Theorem~\ref{thm:tc-high-probability}}
\label{app:tc-example}

The tournament consists of $\ell$ different subtournaments $T_1, \dots, T_{\ell}$, each of which corresponds to a \emph{cyclone} of size $m$, where $m$ is an odd positive integer which is sufficiently larger than $\ell$. 
A cyclone of size $m$ is a tournament in which the $m$ alternatives are arranged on a cycle and each alternative dominates its $(m-1)/2$ successors on the cycle. 
For ease of presentation, each $T_i$ has one distinguished alternative which we call $v_i$. For two alternatives $u$ and $v$ from distinct subtournaments, say $u \in V(T_j)$ and $v \in V(T_{j'})$, it holds in general that $u$ dominates $v$ if and only if $j < j'$. However, there are $\ell - 1$ exceptions: all distinguished alternatives dominate $v_1$, i.e., $(v_2,v_1), (v_3,v_1)\dots, (v_{\ell},v_1) \in E(T)$; we call these \emph{backward edges}.
See \Cref{fig:tc-diversity} for an illustration.

\begin{figure}[!ht]
\centering
\scalebox{0.55}{
\begin{tikzpicture}

\def \n {7}
\def \m {4}
\def \radius {.8cm}
\def \margin {8} 
\foreach \s in {1,...,5}
{ 
      \node[circle,fill,inner sep=2pt](\s) at ({360/\n * (\s - 1)-13}:\radius) {}; 
}

\node[circle,inner sep=2pt](7) at ({275}:\radius) {\small$\dots$}; 

\foreach \s in {6,...,10}
{ 
      \node[circle,fill,inner sep=2pt,xshift=3.5cm](\s) at ({360/\n * (\s - 6)-13}:\radius) {}; 
}

\node[circle,inner sep=2pt,xshift=3.5cm](7) at ({275}:\radius) {\small$\dots$}; 

\foreach \s in {11,...,15}
{ 
      \node[circle,fill,inner sep=2pt,xshift=7cm](\s) at ({360/\n * (\s - 11)-13}:\radius) {}; 
}

\node[circle,inner sep=2pt,xshift=7cm](7) at ({275}:\radius) {\small$\dots$}; 
\foreach \s in {16,...,20}
{ 
      \node[circle,fill,inner sep=2pt,xshift=12.5cm](\s) at ({360/\n * (\s - 16)-13}:\radius) {}; 
}

\node[circle,inner sep=2pt,xshift=12.5cm](7) at ({275}:\radius) {\small$\dots$};

\node[below=.1cm] at (3){\Large$v_1$};
\node[below=.1cm] at (8){\Large$v_2$};
\node[below=.1cm] at (13){\Large$v_3$};
\node[below=.1cm] at (18){\Large$v_{\ell}$};

\node[circle,draw, very thick,inner sep=.8cm] at (0,0)(a1){};  \node[below=1.5cm] at (a1){\Large$T_1$};
\node[circle,draw,very thick,inner sep=.8cm] at (3.5,0)(a2){}; \node[below=1.5cm] at (a2){\Large$T_2$};

\node[circle,draw, very thick,inner sep=.8cm] at (7,0)(a3){}; \node[below=1.5cm] at (a3){\Large$T_3$};

\node at (9.8,0)(d){\Large$\dots$};
\node[circle,draw, very thick,inner sep=.8cm] at (12.5,0)(a4){}; \node[below=1.5cm] at (a4){\Large$T_{\ell}$};

\draw[->,ultra thick] (a1) to (a2);
\draw[->,ultra thick] (a2) to (a3);
\draw[->,ultra thick] (a3) to (d);
\draw[->,ultra thick] (d) to (a4);

\draw[->,thick] (8) to[bend right] (3);
\draw[->,thick] (13) to[bend right] (3);
\draw[->,thick] (18) to[bend right] (3);
\end{tikzpicture}
}

\caption{Illustration of the example showing that $\mov_{\tc}$ does not always follow the formula in Theorem \ref{thm:tc-high-probability}. Each $T_i$ is a ``cyclone'' of size $k$ and has one distinguished alternative~$v_i$.}
\label{fig:tc-diversity}
\end{figure}
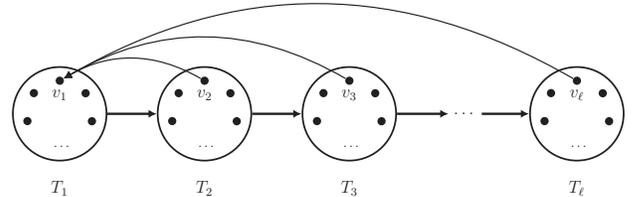

One can check that all alternatives belong to $\tc(T)$.
We claim that $\mov_{\tc}(x,T)=\ell-(i-1)$ if $x \in V(T_i)$ for $i\geq 2$.
Since reversing the backward edges $(v_i,v_1), \dots, (v_{\ell},v_1)$ makes $v_1$ unreachable from $x$, we have $\mov_{\tc}(x,T)\leq\ell-(i-1)$. 
For the other direction, by \Cref{lem:xcut-kkings-destr}, it suffices to show that even if $\ell-i$ edges are removed, $x$ can still reach every other alternative via some directed path.
Suppose that $\ell-i$ edges are removed.
We first claim that in any subtournament $T_j$, every alternative can still reach every other alternative.
Indeed, if the cyclone $T_j$ consists of the alternatives $z_1,\dots,z_m$ in this order, then before the edges are removed, $z_1$ can reach $z_{(m+3)/2}$ via $(m-1)/2$ (disjoint) paths of length two; at least one of these paths remains intact after the edge removal as long as $m > 2\ell+1$.
Similarly, $z_{(m+3)/2}$ can still reach $z_2$, meaning that $z_1$ can also reach $z_2$.
Applying the same argument repeatedly, in $T_j$, every alternative can still reach every other alternative.
Now, since $\ell-i$ edges are removed, one of the backward edges $(v_i,v_1), \dots, (v_{\ell},v_1)$ remains intact, say $(v_s,v_1)$.
By going to $v_s$ via some alternative in $T_s$, our alternative $x$ can then reach $v_1\in T_1$, from where it can also reach all other alternatives in $T$.
Hence $\mov_{\tc}(x,T)=\ell-(i-1)$.
In particular, $\{\mov_{\tc}(x,T) \mid x \in V(T)\}\supseteq \{1,2,\dots,\ell-1\}$.
This example therefore also shows that $\mov_{\tc}$ can take on an arbitrary large number of values in a tournament.

Finally, note that the formula in Theorem \ref{thm:tc-high-probability} predicts a \mov{} value of $ (m-1)/2$ for all alternatives.
This can be made arbitrarily larger than $\ell-1$ by choosing $m$ to be as large as desired.

\end{document}